\documentclass[a4paper,12pt,english]{article}
\usepackage{authblk}

\usepackage[mathscr]{eucal}
\usepackage{amsmath,amsfonts,amssymb,amsthm}
\usepackage{times}
\usepackage{hyperref}

\newcommand{\ruth}[1]{{#1}}

\newtheorem{rem}{Remark}[section]

 \newtheorem{prop}{Proposition}[section]
 
 \newtheorem{definition}[prop]{Definition}

\renewcommand{\tilde}{\widetilde}
\renewcommand{\hat}{\widehat}

\newcommand{\bref}[1]{\textbf{\ref{#1}}}

\renewcommand{\mod}{\,\rm mod \,}
\newcommand{\p}[1]{|#1|}
\newcommand{\gh}[1]{\mathrm{gh}(#1)}

\newcommand{\dx}{\mathrm{d_X}}

\renewcommand{\d}{\partial}
\renewcommand{\dh}{\mathrm{d_h}}

\renewcommand{\geq}{\,{\geqslant}\,}
\renewcommand{\leq}{\,{\leqslant}\,}

\newcommand{\binner}[2]{%
  {\langle}\kern-4.15pt{\langle}#1{,}\,#2{\rangle}\kern-4.15pt{\rangle}}
\newcommand{\commut}[2]{[#1{,}\,#2]}

\newcommand{\ffrac}[2]{\raisebox{.5pt}%
  {\footnotesize$\displaystyle\frac{#1}{#2}$}\kern1pt}

\newcommand{\dl}[1]{\mathchoice{\ffrac{\d}{\d #1}}{\frac{\d}{\d #1}}{\ffrac{\d}{\d #1}}{\ffrac{\d}{\d #1}}}

\newcommand{\cC}{\mathcal{C}}

\newcommand{\fR}{\mathbb{R}}

 \def\cE{\mathcal{E}}

 \def\cI{\mathcal{I}}
 \def\cJ{\mathcal{J}}

\numberwithin{equation}{section}
\usepackage{cleveref}
\usepackage[left=1.5cm,right=1.5cm,
top=1cm,bottom=2cm,bindingoffset=0cm]{geometry}
\newcommand{\ROD}{}
\newcommand{\Co}{\mathrm{C}}
\newcommand{\We}{\mathrm{W}}
\DeclareMathOperator{\sign}{sign}
\newcommand{\dX}{\mathrm{d_X}}
\newcommand{\dJ}{\mathrm{d_{\cJ}}}

\title{Asymptotic symmetries of gravity in the gauge PDE approach\\[20pt]}

\author{Maxim Grigoriev$^{a,b}$ and Mikhail Markov$^{a,b,c}$}
\affil{$^{a}$ Lebedev Physical Institute,\protect\\
  Leninsky Ave. 53, 119991 Moscow, Russia \vspace{1em}
  \\
  $^{b}$ Institute for Theoretical and Mathematical Physics,\protect\\
  Lomonosov Moscow State University, 119991 Moscow, Russia \vspace{1em} \\ 
   $^{c}$ Steklov Mathematical Institute of Russian Academy of Sciences, \protect \\
  Gubkina str. 8, 119991 Moscow, Russia }  \vspace{1em}

\date{}  

\begin{document}

\maketitle
\begin{abstract}
We propose a framework to study local gauge theories on manifolds with boundaries and their asymptotic symmetries, which is based on representing them as so-called gauge PDEs. These objects extend the conventional BV-AKSZ sigma-models to the case of not necessarily topological and diffeomorphism invariant systems and are known to behave well when restricted to submanifolds and boundaries. We introduce the notion of gauge PDE with boundaries, which takes into account generic boundary conditions, and apply the framework to asymptotically flat gravity. In so doing, we start with a suitable representation of gravity as a gauge PDE with boundaries, which implements the Penrose description of asymptotically simple spacetimes. We then derive the minimal model of the gauge PDE induced on the boundary and observe that it provides the Cartan (frame-like) description of a (curved) conformal Carollian structure on the boundary. Furthermore, imposing a version of the familiar boundary conditions in the induced boundary gauge PDE, leads immediately to the conventional BMS algebra of asymptotic symmetries. Finally, we briefly sketch the construction for asymptotically (A)dS gravity.

\end{abstract}

\tableofcontents
\section{Introduction}

Asymptotic symmetries play a prominent role in modern QFT and gravity~\cite{Bondi:1962px, Sachs:1962zza, Penrose:1964ge, Brown:1986nw, Aharony:1999ti, Bagchi:2009my,Barnich:2010eb,Bagchi:2012yk,Strominger:2013jfa,Strominger:2014pwa,Prema:2021sjp}.  They originate from gauge transformations that preserve the boundary conditions imposed on fields but, at the same time, are not to be considered as genuine gauge transformations and hence define new physical symmetries. In particular, asymptotic symmetries critically depend on the choice of the boundary conditions for fields and gauge parameters.

Historically, the first and very influential example of asymptotic symmetries is the Bondi-Metzner-Sachs (BMS) symmetries of asymptotically flat gravity~\cite{Bondi:1962px, Sachs:1962zza}. In contrast to naive expectations, the simplest natural  choice of boundary conditions employed in~\cite{Bondi:1962px, Sachs:1962zza}  results in the enhancement of the Poincar\'e group by the so-called supertranslations, giving the infinite-dimensional symmetry group which is now known as the BMS group. The latter group, or its generalizations associated with different choices of boundary conditions, is now considered to be a proper symmetry group of the gravitational S-matrix and has been related~\cite{Strominger:2014pwa,He:2014laa} to the celebrated soft graviton theorem and the gravitational memory effect.

The proper geometrical setup for studying asymptotic symmetries of gravity was proposed by Roger Penrose, who introduced the notion of asymptotically simple spacetime~\cite{Penrose:1965am}.  More specifically, this is a spacetime $(\tilde M,\tilde g)$ that is diffeomorphic to the interior of the spacetime $(M,g)$ with boundary such that in the interior $g=\Omega^2 \tilde g$ for some smooth function $\Omega$ satisfying $\Omega>0$  and $\Omega|_\cJ=0,d\Omega|_\cJ\neq 0$, where $\cJ=\d M$ is the boundary. In other words, the idea is to realize the  boundary at infinity as the usual boundary of the auxiliary spacetime. More details can be found in e.g.~\cite{Penrose:1972ea,Geroch:1977big,Alessio:2017lps}.

Recent decades have shown an increasing interest in asymptotic symmetries, not only in the context of gravity but also in general gauge theories, including Yang-Mills, topological systems, and higher-spin gauge theories~\cite{Banados:1998gg,Henneaux:2010xg,Campoleoni:2010zq, Gaberdiel:2010pz,Afshar:2013vka,Strominger:2013lka,Campoleoni:2017mbt,Bekaert:2022ipg}. From this perspective, asymptotic symmetries are to be considered as a general feature of gauge theories on manifolds with (asymptotic) boundaries. This calls for a proper gauge-theoretical understanding of asymptotic symmetries. Various approaches are available in the literature. In particular, the first principle understanding of asymptotic symmetries is provided within the Hamiltonian approach~\cite{Brown:1986ed,Brown:1986nw}, see also~\cite{Banados:1998gg}, at the price of manifest covariance. A covariant generalization can be achieved with one or another version of the covariant phase space approach~\cite{Barnich:2001jy,Freidel:2021fxf,Freidel:2021cjp}, see also \cite{Compere:2018aar} and references therein.

A powerful and systematic framework for (quantum) gauge theories is provided by the Batalin-Vilkovisky (BV) formalism~\cite{Batalin:1981jr,Batalin:1983wj} or, more precisely, its modern enhancements, such as the jet-bundle BV approach to local gauge theories, see e.g.~\cite{Barnich:2000zw,Barnich:2001jy}. Of special attention in the present work is the so called BV-AKSZ framework~\cite{Alexandrov:1995kv}, initially proposed in the context of topological models. An interesting feature observed in~\cite{Grigoriev:1999qz,Barnich:2003wj,Barnich:2005ru,Cattaneo:2012qu,Grigoriev:2012xg} (see also \cite{Barnich:2010sw,Grigoriev:2010ic,Grigoriev:2010ic}) is that a BV-AKSZ system naturally induces a shifted AKSZ system on any space-time submanifold. For instance, an AKSZ version of the Hamiltonian BFV formulation is induced on a space-like submanifold of spacetime. By combining this observation with the construction of~\cite{Barnich:2010sw}, see also~\cite{Grigoriev:2010ic,Grigoriev:2012xg}, which allows one to reformulate a general local gauge system as an AKSZ-like model, one arrives at the framework to analyze boundary values and symmetries of generic local gauge systems. This approach has been successfully employed in~\cite{Bekaert:2012vt,Bekaert:2013zya,Chekmenev:2015kzf} in the study of boundary values of generic gauge fields on AdS space, see also~\cite{Vasiliev:2012vf} for a related approach, and in the reconstruction of bulk theories from the boundary data~\cite{Bekaert:2017bpy,Grigoriev:2018wrx}. Note that it does not employ the symplectic structure of BV formalism and is applicable to non-Lagrangian systems or systems whose Lagrangian is not specified.  Let us also mention the somewhat related approach of~\cite{Cattaneo:2012qu} to Lagrangian gauge systems on manifolds with boundaries, see also~\cite{Mnev:2019ejh,Rejzner:2020xid,Riello:2023ptb}.

In this work we develop an approach to gauge theories on manifolds with boundaries and their asymptotic symmetries, which is based on representing a given local gauge theory as a so-called \textit{gauge PDE}. Gauge PDE (gPDE)
is a generalization of the non-Lagrangian BV-AKSZ  formulation to the case of general gauge theories. Although the term gauge PDE and its geometrical definition was introduced only in~\cite{Grigoriev:2019ojp}, the framework was originally put forward already in~\cite{Barnich:2010sw}, see also~\cite{Barnich:2004cr}, under the name \textit{parent formulation}. Just like AKSZ systems, gPDEs behave well with respect to the restriction to space-time submanifolds and hence provide a natural framework to study gauge theories on manifolds with boundaries. More precisely, a gPDE is a bundle over spacetime and its pullback to a submanifold is again a gPDE, see e.g.~\cite{Grigoriev:2022zlq} and references therein. Gauge PDEs can be also considered as a BV-BRST extension and generalization of the so-called unfolded formalism developed in the context of higher spin theories~\cite{Vasiliev:1988xc,Vasiliev:2005zu}. 

We propose the notion of gPDEs with boundaries, which takes into account boundary conditions on fields and gauge parameters. More precisely, the boundary conditions are described by a sub-gPDE of the induced boundary gPDE which is, by definition, the initial gPDE pulled back to the boundary. In these terms asymptotic symmetries can be defined in a rather general and purely geometrical way, giving a systematic description of such systems and their symmetries in terms of differential graded geometry.

The approach is applied to asymptotically flat gravity and is shown to reproduce celebrated BMS symmetries once a gauge PDE version of the well-known boundary conditions is taken. A crucial point of the construction is the gauge theoretical implementation of the Penrose asymptotically simple spacetime. This is achieved by introducing a Weyl compensator field $\Omega$. In so doing the metric sector of the system can be considered as that describing conformal geometry, which allows to employ the equivalent reduction~\cite{Boulanger:2004eh} (see also \cite{Dneprov:2022jyn}) known in the context of conformal gravity. This later step leads to a remarkably simple boundary system which also resembles the conformal-geometry  approach~\cite{Herfray:2020rvq,Herfray:2021xyp,Herfray:2021qmp} to BMS symmetries and more general boundary calculus~\cite{Gover:2011rz,RodGover:2012ib}.

The paper is organized as follows: in Section \bref{sec:gpde-wb}, we briefly recall the gauge PDE approach to local gauge theories and propose its extension to theories on manifolds with boundaries and generally nontrivial boundary conditions. We then define asymptotic symmetries in this setup. In Section \bref{sec:GR-as-sim}, we present a reformulation of gPDE for general relativity in a form convenient for studying the asymptotic structure of this theory. The form is inspired by the Penrose's notion of an asymptotically simple spacetime. In Section~\bref{sec:bound+asympt} we derive the induced boundary system and analyze boundary conditions and asymptotic symmetries in the asymptotically flat case. This involves derivation of a concise minimal model of the induced boundary system, which makes manifest the Carrollian geometry structure. Finally, we sketch the construction in the the case of nonzero cosmological constant and present the respective minimal model.


\section{Gauge PDEs with boundaries and their symmetries}
\label{sec:gpde-wb}

\subsection{Gauge PDEs}

In the approach we employ in this work, local gauge theories, considered at the level of equations of motion, are encoded in the geometrical objects called gauge PDEs (or gPDE for short). The geometry underlying

gPDE can be seen as a generalization of the jet-bundle non-Lagrangian BV formalism\footnote{The version of classical BV formalism at the level of equations of motion was suggested in~\cite{Barnich:2004cr}, see also~\cite{Barnich:2009jy,Kaparulin:2011xy,Sharapov:2016sgx}.} \ruth{,} whose underlying geometrical structure is a jet-bundle of a graded fiber bundle equipped with the BV-BRST differential, see e.g.~\cite{Barnich:2000zw} for more details and references. 

We first need to briefly recall the necessary prerequisites. More detailed exposition can be found in~\cite{Grigoriev:2019ojp}.
\begin{definition}
    A $Q$-manifold (also called dg-manifold) is a $\mathbb{Z}$-graded supermanifold equipped with a homological vector field $Q$, i.e. a vector field of degree 1 satisfying $Q^2=\dfrac{1}{2}[Q,Q]=0$, $\gh Q=1$, $\vert Q\vert=1$, where $\gh\cdot$ denotes $\mathbb{Z}$-degree (often called ghost degree), and $\vert\cdot\vert$ denotes Grassmann parity.
\end{definition}
In this work we only deal with bosonic systems and hence one can simply assume $\vert f \vert=\gh{f}\mod{2}$  for any homogeneous functions, form, vector field, etc. Of course, the framework extends to systems with fermions in a standard way.

The standard simplest example of a $Q$ manifold is a shifted tangent bundle $T[1]X$ over a smooth manifold $X$. Its algebra of functions is just the algebra of differential forms on $X$. Under this identification the de Rham differential corresponds to a homological vector field $\dx$ on $T[1]X$. If $x^\mu$
are local coordinates on $X$ and $\theta^\mu$ the associated coordinates on the fibers, $\dx\equiv \theta^{\mu}\frac{\partial}{\partial x^{\mu}}$.

Let us also recall the definition of a $Q$-bundle, i.e. a  fiber bundle in the category of $Q$-manifolds:
\begin{definition}\cite{Kotov:2007nr}
    1. $Q$-bundle $\pi:(M,Q)\rightarrow(N,q)$ is a locally trivial bundle of graded manifolds $M$ and $N$ such that $\pi^{*}\circ q=Q\circ\pi^{*}$.

    2. A section $\sigma:N\rightarrow M$ is called a $Q$-section if $q\circ\sigma^{*}=\sigma^{*}\circ Q$.

    3. A $Q$-bundle $\pi:(M,Q)\rightarrow(N,q)$ is called locally trivial (as a $Q$-bundle) if it's locally isomorphic to a direct product of the base $(N,q)$ and the fiber  $(F,q^{\prime})$ in such a way that $Q=q+q^{\prime}$, i.e. $Q$ is locally the direct product $Q$-structure.
\end{definition}

There is a natural notion of equivalence for $Q$ manifolds, which, roughly speaking, corresponds to elimination of contractible pairs. From gauge-theoretical viewpoint, such contractible coordinates correspond to auxiliary fields, pure gauge variables and their associated ghosts/antifields. More precisely:
\begin{definition}
    1. A contractible $Q$-manifold is a $Q$-manifold of the form $(T[1]W,d_W)$, where $W$ is a graded vector space considered as a graded manifold, and $d_W$ is the de Rham differential on $T[1]W$.

    2. A $Q$-manifold $(N,q)$ is called an equivalent reduction of $(M,Q)$ if $(M,Q)$ is a locally trivial $Q$-bundle over $(N,q)$ admitting a global $Q$-section, and the fibers of this bundle are contractible $Q$-manifolds.
\end{definition}
Equivalent reduction generates the notion of equivalence. In particular, cohomology of natural complexes, e.g. $Q$-cohomology in differential forms on $E$, multivector fields, etc. on equivalent $Q$-manifolds are isomorphic. 

Locally, the statement that $(N,q)$ is an equivalent reduction of $(M,Q)$ implies that, seen as a $Q$-manifold, $(M,Q)$ is a direct product of $(N,q)$ and a contractible $Q$-manifold. A useful way~\cite{Barnich:2004cr} to identify an equivalent reduction in practice  is to find independent functions $w^a$ such that $Qw^a$ are independent functions as well. It then follows that at least locally a submanifold defined by $Qw^a=0$ and $w^a=0$ is an equivalent reduction of the initial $Q$-manifold. It follows one can find functions $\phi^i$ such that $w^a,v^a=Qw^a,\phi^i$ form a local coordinate system and  $Q\phi^i=Q^i(\phi)$.
In this case, $w^a$ and $v^a$ are standard contractible pairs known in the context of local BRST cohomology, see e.g.~\cite{Brandt:1996mh} and reference therein.

The above notions of equivalent reduction and of equivalence extend to $Q$-bundles over the same base: 
\begin{definition}
Let $(M^\prime,Q^\prime)$  and $(M,Q)$ are $Q$-bundles over the same base $(N,q)$. $(M,Q)$ is called an equivalent reduction of $(M^\prime,Q^\prime)$ if $(M,Q)$ is a locally trivial $Q$-bundle over $(M,Q)$  such that the projection and the local trivializations maps are compatible with projections to $(N,q)$ (i.e. it is a bundle in the category of bundles over $(N,q)$), $(M^\prime,Q^\prime)$ admits a global $Q$-section, and, moreover, the fiber is a contractible $Q$-manifold.
\end{definition}
This generates an equivalence relation for $Q$-bundles. Again, a practical way to identify an equivalent reduction is to find functions $w^a$ such that $w^a,Qw^a$ are independent functions that remain independent when restricted to a fiber (i.e. they can be taken as a part of a fiber coordinate system). It follows that at least locally the subbundle of $(M^\prime,Q^\prime)$ singled out by $w^a=0$ and $Qw^a=0$ is an equivalent reduction.~\footnote{Strictly speaking, in the infinite-dimensional case one should also require the existence of complementary fiber coordinates $\phi^i$ such that $Q\phi^i=Q^i(\phi)$. See~\cite{Barnich:2010sw} for more details.}

Finally, we are ready to formulate the definition of Gauge PDEs.
 \begin{definition}\label{equivred}
    1. Gauge PDE $(E,Q,T[1]X)$ is a $Q$ bundle $\pi:(E,Q)\rightarrow(T[1]X,\dX)$, where $X$ is a real manifold (independent variables). In addition it is assumed that $(E,Q,T[1]X)$ is locally equivalent to nonnegatively graded $Q$-bundle. Moreover, it should be equivalent to a jet-bundle BV-formulation seen as $Q$-bundle over $T[1]X$ with $Q=\dh+s$, where $s$ is the BV-BRST differential and $\dh$ the horizontal differential.\\
        2. Two gauge PDEs over $T[1]X$ are considered equivalent if they are equivalent as $Q$-bundles.
\end{definition}

Gauge PDEs encode local gauge theories. In particular, field configurations are identified with their sections while equations of motion arise as differential conditions on sections. More precisely, 
section $\sigma: T[1]X\rightarrow E$ is a solution to $(E,Q,T[1]X)$ if 
\begin{align}\label{predv-solutions}
        \dx\circ\sigma^{*}=\sigma^{*}\circ Q
    \end{align}
Infinitesimal gauge transformations of the section $\sigma$ are defined as    
\begin{align}\label{predv-gaugetransf}
        \delta\sigma^{*}=\dx\circ \chi_{\sigma}^{*}+ \chi_{\sigma}^{*}\circ Q,
    \end{align}
    where $\chi_{\sigma}^{*}: \cC^{\infty}(E)\rightarrow \cC^{\infty}(T[1]X)$ is of degree $-1$, satisfies
    \begin{align} \label{predv-xi-tr}
        \chi_{\sigma}^{*}(fg)=\chi_{\sigma}^{*}(f)\sigma^{*}(g)+(-1)^{\vert f\vert}\sigma^{*}(f)\chi_{\sigma}^{*}(g)\,, \qquad \forall f,g \in \cC^{\infty}(E)\,, 
        \end{align}
and $\chi_{\sigma}^*(\pi^*(h))=0$ for all $h\in \cC^{\infty}(T[1]X)$. The map $\chi_{\sigma}^*$ is interpreted as a gauge parameter. It is easy to check that the above gauge transformation is an infinitesimal symmetry of the equations of motion~\eqref{predv-solutions}. In a similar way one defines gauge for gauge symmetries.

It is often convenient to parameterize $\chi_{\sigma}^{*}$ in terms of a vertical vector field $Y$ on $E$ of degree $-1$: $\chi_{\sigma}^{*}=\sigma^{*}\circ Y$. It is easy to check that for this choice $\chi_{\sigma}^{*}$ \eqref{predv-xi-tr} is indeed satisfied. Using this representation the gauge transformation of $\sigma$ can be written as $\delta \sigma^*=\sigma^*\circ \commut{Q}{Y}$. Note that a vector field $V\equiv[Q,Y]$ is an infinitesimal symmetry of $(E,Q,T[1]X)$ because it preserves $Q$, the degree, and the bundle structure.

In the case of diffeomorphism-invariant systems, for instance gravity, their gPDE description 
usually requires the additional condition on the allowed class of sections. More precisely, the fiber coordinates typically involve a subset of ``diffeomorphism ghosts''
$\xi^a$, $a=0,\ldots \dim{X}-1$, $\gh{\xi^a}=1$ and sections are restricted by the condition that $e^a_\mu(x)$ defined via $\sigma^*(\xi^a)=e^a_\mu(x) \theta^\mu$, are invertible. This is of course a gPDE counterpart of the familiar condition in the frame-like formulation of gravity. All the systems considered in this work are of this type and the nondegeneracy condition on sections is assumed in what follows.

To complete the discussion of gPDEs let us note that gPDE automatically determine a nonlagrangian jet-bundle formulation of the underlying gauge system. This is induced on the bundle of super-jets of $E$ and its BV-BRST differential is the vertical part of the prolongation of $Q$ to the super-jet bundle. More details can be found in~\cite{Grigoriev:2019ojp,Grigoriev:2022zlq}, see also~\cite{Barnich:2010sw} for the original construction and local proof.

\subsection{Gauge PDEs with boundaries}

Let $X$ be a space-time manifold but now we assume that $X$ has a nontrivial boundary $\Sigma=\d X$ and let $i:\Sigma \to X$ denotes the embedding of the boundary.  Suppose we are given with a gPDE $(E,Q,T[1]X)$ on $X$. This induces a new gPDE $i^* E$ on $\Sigma$ given by a pullback of $E$ to $T[1] \Sigma \subset T[1]X$ (here by a slight abuse of notation $i$ also denotes an induced pushforward $T[1] \Sigma \to T[1]X$). It is easy to check that this is again a gPDE (e.g. by regarding it as a $Q$-submanifold of $E$), which we call induced boundary gPDE. $i^* E$ can be considered as a gPDE describing a gauge theory of unconstrained boundary values of the fields encoded in $(E,Q,T[1]X)$, see~\cite{Grigoriev:2022zlq} for more details and \cite{Bekaert:2012vt,Bekaert:2013zya} for the earlier and less general construction and applications in the context of higher spin holography.

Now we are interested in the gPDE description of systems with possibly nontrivial boundary conditions. We have the following:
\begin{definition} \label{def:BgPDE}
By a gauge PDE with boundaries we mean the following data: $(E,Q,T[1]X,E_\Sigma,T[1]\Sigma)$, where gPDE $(E,Q,T[1]X)$ is a gPDE on $X$ and $(E_\Sigma, Q_\Sigma, T[1]\Sigma)$ is a gPDE on the boundary $\Sigma=\d X$, which is a sub-gPDE of $i^*E$.  In particular, $Q_\Sigma$ is a restriction of $Q$ to $E_\Sigma\subset i^*E \subset E$. 
\end{definition}
Gauge PDE $(E_\Sigma,Q_\Sigma, T[1]\Sigma)$, which is a part of the above definition, can be regarded as a \textit{gPDE of boundary conditions}. For instance, if $E_\Sigma=i^*E$ this means that no boundary conditions are imposed. General situations are described by nontrivial $Q$-subbundles of $i^*E$. It is important to stress that in general, $E_\Sigma$ restricts not only the boundary values of fields but also the boundary values of gauge parameters and parameters of gauge-for-gauge symmetries (if any). 
\begin{rem}\label{rem:reduction}
Even if $E_\Sigma$  doesn't coincide with $i^*E$ it doesn't necessarily mean that we are dealing with nontrivial boundary conditions. This happens if $E_\Sigma$ is an equivalent reduction of $i^*E$ in which case $E_\Sigma$ implements elimination of auxiliary fields and pure gauge variables.
\end{rem}

Although the above definition is quite general, in the context of asymptotic symmetries it is useful to allow $(E,Q,T[1]X)$ to be \ruth{} slightly locally nontrivial. Namely, $E$ restricted to the interior of $X$ and $i^*E$ are still required to be locally trivial while the typical fiber of $i^* E$ can differ from the fiber over the interior by a subset of measure zero. For our present purposes it is enough to allow the fiber of $i^*E$ to be a manifold with boundary whose interior coincides with the typical fiber over the interior. In this case the total space of $E$ restricted to the interior of $X$ is itself the interior of a manifold with corners. Restricting it to $\d X$ gives a total space of $i^*E$ whose fiber is a manifold with boundary. As we are going to see, at more practical level we actually work with $i^*E$ and its sub-gPDE $E_\Sigma$ which are locally trivial. Note that a gPDE with boundary could be defined in terms of a single locally-nontrivial bundle $E^\prime \to T[1]X$ such that its fibers over the boundary shrinks to those of $E_\Sigma$ but we prefer to keep boundary conditions explicit.

The field theoretical interpretation of the above definition becomes clear with the help of:
\begin{definition}
1. A solution of $(E,Q,T[1]X,E_\Sigma,T[1]\Sigma)$ is a section $\sigma: T[1]X \to E$ satisfying $\sigma^* \circ Q= \dx \circ \sigma^*$ and such that its restriction to $T[1]\Sigma$ belongs to $E_\Sigma$, i.e. is a solution to $(E_\Sigma,Q_\Sigma, T[1]\Sigma)$.

2. A gauge parameter is a vertical vector field $Y$ on $E$ such that $\gh{Y}=-1$ and its restriction to $i^* E$ is tangent to $E_\Sigma \subset i^* E$. In other words, gauge parameters should satisfy the boundary conditions encoded in $E_\Sigma\subset i^* E$.

3. A gauge transformation of section $\sigma$ is defined as
\begin{equation}
\label{gs}
\delta_Y \sigma^* = \dx \circ  \sigma^* \circ Y + \sigma^* \circ Y \circ Q   \,. 
\end{equation}
\end{definition}
The following comments are in order: it is easy to check that if $\sigma$ is a solution then $\sigma+\delta_Y \sigma$ with $\delta_Y\sigma$ determined by~\eqref{gs}, is again a solution (to first order in $Y$).  Moreover, in this case the gauge transformation~\eqref{gs} can be rewritten as:
\begin{equation}
\label{gs2}
\delta_Y \sigma^* = \sigma^* \circ \commut{
Q}{Y}\,.    
\end{equation}
Restricting \eqref{gs} to $T[1]\Sigma$, one finds a standard gauge transformation for $(E_\Sigma,Q_\Sigma, T[1]\Sigma)$ whose parameter is $Y$ restricted to $E_\Sigma\subset  i^*E \subset E$ (recall that $Y$ is vertical and hence is tangent to $i^*E$, while the above definition requires $Y$ to be tangent to $E_\Sigma$). Gauge-for-gauge transformations can be defined in a similar way. In particular, parameters of the gauge-for-gauge transformations of stage 1 are vertical vector fields of degree $-2$ and tangent to $E_\Sigma$.

All the above definitions can be generalised to the case where $\Sigma$ is a generic submanifold. This version can be relevant in describing theories with defects. Generalization to the case of manifolds with corners or higher codimension strata is also possible.


\subsection{(Asymptotic) symmetries in gPDE terms}
\label{sec:as-sym-gPDE}

Let us now turn to the discussion of symmetries. Given a gauge PDE, an \textit{infinitesimal symmetry} is by definition a vector field $W$ which preserves all the structures, i.e. $\commut{Q}{W}=0$, $W$ is vertical, and $\gh{W}=0$ (though symmetries of the \ruth{} nonvanishing ghost number are also of interest).  The infinitesimal transformation of a solution $\sigma$ under a symmetry transformation determined by $W$ is defined to be: 
\begin{equation}
    \delta_W \sigma^*=\sigma^*\circ W\,.
\end{equation}
It is easy to check that it defines an infinitesimal symmetry transformation that takes solutions to solutions.
Gauge symmetries are the ones where $W=\commut{Q}{Y}$, where $Y$ is a gauge parameter. In particular, in this case the above transformation coincides with~\eqref{gs2} so that it is natural to regard symmetries of the form $W=\commut{Q}{Y}$ as trivial. As we only consider infinitesimal symmetries, in what follows we systematically omit "infinitesimal".

\textit{Inequivalent symmetries} (also known as \textit{global} or  \textit{physical})  can be defined as the respective quotient of all symmetries modulo the ideal of the gauge ones and hence are given by $Q$-cohomology in vertical vector fields. One can check that in the case of usual jet-bundle BV formulation of a local gauge theory, this definition reproduces the standard one, at least locally. Details can be found in Appendix~\bref{sec:gpdesym}. 

As far as gauge PDEs with boundaries are concerned, a natural definition of a symmetry is that it is a vertical vector field $W$ such that  $\commut{Q}{W}=0$ and $W$ restricted to $i^*E$ is tangent to $E_\Sigma \subset i^*E$. At the same time, genuine gauge symmetries in this case are those symmetries of the form $W=\commut{Q}{Y}$ whose  parameters $Y$ are tangent to $E_\Sigma$. All the above discussion applies to symmetries of arbitrary definite ghost degree.

Let now $(E,Q,T[1]X,E_\Sigma,T[1]\Sigma)$ be a gPDE with boundaries. A common lore is that asymptotic symmetries are gauge symmetries of the system (defined as if there were no boundary conditions) that preserve boundary conditions while those whose parameters satisfy  boundary conditions for gauge parameters are genuine gauge symmetries and should be considered trivial. In the case of Lagrangian systems the extra requirements are to be imposed.  In the gPDE setup this can be formalised as follows:
\begin{definition}
Asymptotic symmetries of $(E,Q,T[1]X,E_\Sigma,T[1]\Sigma)$ are symmetries of $(E_\Sigma,Q_\Sigma,T[1]\Sigma)$,
which are restrictions to $E_\Sigma$ of those gauge symmetries of $i^* E$ that preserve $E_\Sigma \subset i^* E$. Gauge symmetries of $(E_\Sigma,Q_\Sigma,T[1]\Sigma)$, i.e. 
vector fields on $E_\Sigma$ of the form $\commut{Q}{Y}|_{E_\Sigma}$ with $Y$ tangent to $E_\Sigma$,
are considered trivial asymptotic symmetries.
\end{definition}
More explicitly, asymptotic symmetries are vertical vector fields that are tangent to $E_\Sigma$
and have the form $\commut{Q}{Y}$ with $Y$ vertical. These are considered modulo vector fields vanishing on $E_\Sigma$. Moreover, vector fields $\commut{Q}{Y}$ with $Y$ tangent to $E_\Sigma$ are genuine gauge symmetries and are therefore trivial. Asymptotic symmetries form a subalgebra of all symmetries of $(E_\Sigma,Q_\Sigma,T[1]\Sigma)$.

An alternative would be to define asymptotic symmetries as all symmetries of $(E_\Sigma,Q_\Sigma,T[1]\Sigma)$ modulo its own gauge symmetries.  Another alternative is to require asymptotic symmetries to arise as restrictions to $E_\Sigma$ of symmetries of $E$, but this does not seem to make any difference if we restrict ourselves to local analysis, as we do in this work.

Later on we also need a slightly more general framework applicable to the case where $E_\Sigma$ is not a regular submanifold of $i^*E$ while its prolongation to the bundle of jets of its sections is. This occurs in applications in which a frame field arises as a component of $\sigma^*(\xi^a)$, where $\sigma:T[1]\Sigma \to i^*E$, and therefore the condition that the frame field is invertible cannot be implemented in terms of the fiber geometry and is imposed instead as a condition on sections.

Such formulations arise in practice if one is after a concise formulation of the boundary gPDE. To cover this situation one allows $Y$ to be a generalized vector field, i.e. its coefficients are allowed to depend on jets of sections. If $\cI_{\Sigma}$ is the ideal of $E_\Sigma \subset i^*E$ then the condition that $\commut{Q}{Y}$ is tangent to $E_\Sigma$ is replaced by
\begin{equation}
    \dx \sigma^*(Yf)+\sigma^*(YQf)=0 \quad \forall f \in \cI_\Sigma\,.
\end{equation}
This should hold for all sections of the gPDE of boundary condition, i.e.  sections of $i^*E$ such that $\sigma^*(\cI_\Sigma)=0$. The above condition ensures that the corresponding symmetry transformation sends solutions to $E_\Sigma$ to themselves. Note that if one doesn't want to employ such a generalization it is always possible to work in terms of the associated parent gPDE whose underlying bundle is the super-jet bundle of $i^*E$
and hence the prolongation of $E_\Sigma$ is a smooth submanifold, see~\cite{Barnich:2010sw,Grigoriev:2019ojp} for further details of parent gPDEs.


\subsection{Asymptotic symmetries in the presymplectic gPDE framework}

In the case of Lagrangian systems, asymptotic symmetries can be defined as gauge transformations whose associated charges become nontrivial due to boundary conditions. Althogh in this work we restricted ourselves to the analysis at the level of equations of motion, let us briefly comment on how such an approach can be implemented in the gPDE framework. 

A Lagrangian system can be described by a gPDE $(E,Q,T[1]X)$ equipped with the compatible presymplectic structure  $\omega$ of degree $n-1$, $n=\dim{X}$, such that:
\begin{equation}
d\omega=0\,, \qquad L_Q\omega \in \cI \,, \qquad i_Q L_Q\omega\in \cI\,,
\end{equation}
where $\cI$ denotes the ideal of forms on $E$ generated by the forms of positive degree, pulled back from the base, i.e. by $dx^\mu,d\theta^\mu$ in standard coordinates. We also fix a presymplectic potential $\chi$ such that $\omega=d\chi$. Note that for $n>1$ it exists globally as the respective de Rham cohomology is empty. More details on presymplectic gPDEs can be found in~\cite{Grigoriev:2022zlq,Dneprov:2022jyn,Grigoriev:2020xec}, see also \cite{Grigoriev:2016wmk,Grigoriev:2021wgw,Alkalaev:2013hta} for earlier relevant works.

We say that boundary gPDE $E_\Sigma$ and symplectic potential $\chi$ are compatible if the pullback of $\chi$ to $E_\Sigma$ vanishes. It turns out that this is a sufficient condition for the respective action to be differentiable. Indeed, the above data defines a presymplectic AKSZ-like action (also known as intrinsic action):
\begin{equation}
 S[\sigma]=\int_{T[1]X} \sigma^*(\chi)(\dx)+
 \sigma^*(H)\,,
\end{equation}
where the ``covariant Hamiltonian/BRST charge"  $H$ is defined through $i_Q\omega+dH\in \cI$. Note that picking an equivalent $\chi=\chi+d\alpha$ doesn't affect equations of motion but adds a boundary term $\int\dx \sigma^*(\alpha)$ to the above action.  Consider a variation of $S[\sigma]$ under $\sigma \to \sigma+\delta\sigma$. Representing $\delta\sigma$ as $\sigma^* \circ V$, where $V$ is a vertical vector field on $E$, one finds that the boundary term has the form:
\begin{equation}
\delta S=\int \text{``EOM''}\,\, \delta\sigma+\int \sigma^*(i_V \chi)   
\end{equation}
Because $\delta \sigma$ should preserve boundary conditions, $V$ is tangent to $E_\Sigma$ so that boundary contribution vanishes provided  $\chi$ and $E_\Sigma$ are compatible as we assume in what follows.

Let us now turn to conserved currents (conservation laws) associated to symmetries. A symmetry $W$ is called compatible with presymplectic structure if 
\begin{equation}
\label{omega-cons}
 L_W \omega+L_Q d\alpha \in \cI   \,,
\end{equation}
for some $1$-form $\alpha$ of ghost degree $\gh{\alpha}=\gh{W}+n-2$.  Given a compatible symmetry one can define an associated generalised conserved current, which is a degree $\gh{W}+n-1$ function defined through:~\footnote{The discussion of global symmetries and conserved currents in the presymplectic BV-BRST approach can be found in~\cite{Sharapov:2016sgx}.}
\begin{equation}
\label{H-def}
i_W \omega -(-1)^{\p{W}}L_Q\alpha-d H_W\in \cI\,,
\end{equation}
The consistency condition $d(i_W \omega  -(-1)^{\p{W}} L_Q\alpha)\in \cI$ holds thanks to \eqref{omega-cons}. If $\alpha$ is fixed  $H_W$ is defined modulo functions of the form $\pi^*(f), f\in \cC^{\infty}(T[1]X)$. Moreover,  it follows $d(QH_W)\in \cI$ and hence $Q H_W= \pi^*(h)$ for some function $h$ on $T[1]X$ so that by adding to $H_W$ a function of the form $\pi^*(f)$ one can achieve $QH_W=0$. This defines a map from compatible symmetries to conserved currents.

Given a conserved current, i.e. a $Q$-closed function $H_W$ on $E$, one can define the respective charge as
\begin{equation}
\mathbf{H_W}[\sigma]= \int_{T[1]C}\sigma^*(H_W)\,,  
\end{equation}
where $\sigma$ is a solution ($Q$-section) of $E$ restricted to a shifted tangent bundle of submanifold $C \subset X$ of codimension $1-\gh{W}$. Note that adding a $Q$-exact piece to $H_W$ results in the addition of a $\dx$-exact term in the integrand and hence this only  contributes to the boundary term (if $C$ has a nontrivial boundary). The above charge doesn't change under deformations of $C$ provided $\d C$ is kept undeformed, because $\dx \sigma^*(H_W)=\sigma^*(QH_W)=0$ if $\sigma$ is a solution. 

If $W$ is a gauge symmetry, i.e. $W=\commut{Q}{Z}$ with $Z$ vertical it is automatically compatible with $\omega$ because  $L_{\commut{Q}{Z}} \omega  = -(-1)^{\p{Z}}L_Q d L_Z\chi+\cI$. The associated conserved current determined by the above map is $Q$-exact and can be taken in the form $H_W=Q (i_Z\chi)$. In particular, in the case of a gPDE with boundary the currents associated to gauge symmetries (i.e. with $Z$ tangent to $E_\Sigma$) necessarily vanish on the boundary provided $\chi$ and $E_\Sigma$ are compatible. In other words charges associated to genuine gauge symmetries vanish while those associated to asymptotic symmetries are generally nontrivial.


\section{GR as a gauge PDE}\label{sec:GR-as-sim}

\subsection{Off-shell GR as a gauge PDE}

We start by reformulating Riemannian geometry as a local gauge theory or, more precisely, a gauge PDE. This system can also be seen as an off-shell gravity, i.e. a gauge theory whose fields are components of the metric and gauge transformations act as diffeomorphisms. 

In the gPDE language the underlying bundle is given by:
\begin{equation}
    \cE \to X\,, \qquad  \cE=(T^*X \vee T^*X)_{\mathrm{nd}}\oplus T[1]X\,.
\end{equation}
Sections of the first summand are metrics $\tilde g_{ab}(x)$ while coordinates on the fibers of the second summand are diffeomorphism ghosts $\xi^a$. The desired gPDE $\tilde E \to T[1]X$ can be taken to be $J^\infty(\cE) \to X$, pulled back to $T[1]X$ by the canonical projection $T[1]X\to X$. In plain words, the fiber coordinates are $D_{(a)}\tilde{g}_{bc}, D_{(a)}\xi^{b}$, where $D_a$ denote canonical total derivative in $J^\infty(\cE)$, $(a)$ denotes a symmetric multi-index, and the degree is assigned in a standard way: $\gh{D_{(a)}\tilde{g}_{bc}}=0, \gh{D_{(a)}\xi^{b}}=1$.  In terms of local coordinates the $Q$ structure is determined by
\begin{align}\label{Q-riem}
     Qx^\mu=\theta^{\mu},\qquad Q \tilde{g}_{bc}=\xi^a D_a \tilde{g}_{bc}+\tilde{g}_{ac}D_{b}\xi^a+\tilde{g}_{ba}D_c \xi^a,\qquad
    Q\xi^b=\xi^a D_a \xi^b,
\end{align}
and $[Q,D_a]=0$. Note that $\tilde{g}_{bc}=\tilde{g}_{cb}$ and $ \tilde g_{bc}$ is invertible.  Note also that here we use generic coordinates $x^\mu$ on the base $X$. To be more specific, once the jet bundle and the action of $Q$ on the fiber is defined in terms of a fixed coordinate system $x^a$ we have a freedom of using any coordinate system $x^\mu$ on the base. This happens because $Q$ is locally a product $Q$-structure of $\dx$ and the $Q$-structure of the typical fiber or, in other words, the underlying bundle is locally trivial as a $Q$-bundle.

In what follows we often refer to $\tilde g_{ab}$ as a metric. Similarly, we refer to the Cristoffel symbols, Riemann curvature, etc.  seen as respective functions in $D_{(a)}\tilde{g}_{bc}$ just as Christoffel symbols, Riemann curvature, etc. This is natural because such local functions coincide with the respective objects if one evaluates them on the prolongation of a section $\sigma_0\,:X\, \to \cE$. However, from the gPDE point of view, this only happens in a particular gauge~\eqref{metric-gauge}.

\begin{rem}
The above gPDE is not exactly the standard BV-BRST jet-bundle equipped with the horizontal differential $\dh$ and the BRST differential $\gamma$, see e.g.~\cite{Barnich:1995ap,Barnich:1995db}. Although the action of $Q$ on fiber variables coincides with that of the standard BV-BRST differential it actually corresponds to the total BRST differential $\dh+\gamma$. More precisely, thanks to the diffeomorphism invariance one can bring $\dh+\gamma$ to the form~\eqref{Q-riem} by a change of fiber coordinates, see e.g.~\cite{Barnich:2010sw} for more details. 
\end{rem}

A local proof of the equivalence of the above gPDE and the standard jet-bundle BV-BRST formulation of off-shell GR can be found in~\cite{Barnich:2010sw}. In any case, it is not difficult to explicitly consider solutions and gauge transformations and check that we are indeed dealing with the off-shell GR. Probably the simplest way to see the equivalence is to observe that the gauge condition
\begin{equation}
\label{metric-gauge}
\sigma^*(\xi^a)=\theta^a, \qquad \sigma^*(D_{(b)} \xi^a)=0
\end{equation}
is reachable locally. In this gauge the remaining equations of motion simply tell us that $\tilde{g}_{ab}$ is unconstrained and that $\sigma^*(D_{(a)} \tilde{g}_{bc})=\d_{(a)} \sigma^*(\tilde{g}_{bc})$. Moreover, in this gauge the residual gauge parameters $\chi^*(D_{(a)}\xi^b)$ are all determined by $\epsilon^a(x)=\chi^*(\xi^a)$ and the residual gauge transformation of $\sigma^*(\tilde{g}_{ab})$ is given by $L_{\epsilon}\sigma^*(\tilde{g}_{ab})$ so that indeed we are dealing with off-shell gravity. More detailed discussion of the analogous gauges in a more general context can be found in~\cite{Basile:2022nou}.

The above system can be equivalently extended to provide a gauge-theoretical implementation of the Penrose description of asymptotically simple spaces~\cite{Penrose:1962ij,Penrose:1964ge}. More specifically, we extend the fiber of $\cE$ with extra coordinates $\Omega$, $\Omega>0$ and $\lambda$, with $\gh{\Omega}=0, \gh{\lambda}=1$ and extend $Q$ as follows:
\begin{align} \label{riem-omega}
    Q\Omega=\xi^{a}D_{a}\Omega+\lambda\Omega,\qquad Q\lambda=\xi^a D_a\lambda,\qquad \commut{D_a}{Q}=0 \,.
\end{align}

Condition $\Omega>0$ is crucial in ensuring the equivalence of the initial and the extended system in the sense of \bref{equivred}. Thanks to this condition we can introduce new coordinates $g_{bc}\equiv \Omega^2 \tilde{g}_{bc}$. In these coordinates the action of $Q$  is given by:
\begin{align}
\label{qext}
        Qx^\mu=\theta^{\mu},\quad Q g_{bc}=\xi^a D_a g_{bc}+g_{ac}D_{b}\xi^a+g_{ba}D_c \xi^a+2\lambda g_{bc},\quad
    Q\xi^b=\xi^a D_a \xi^b\,.
    \end{align}
\begin{definition}\label{def-confgeom}  
The extended system with $\Omega>0$ and the $Q$ structure determined by~\eqref{qext},\eqref{riem-omega} is called conformal-like off-shell gravity.
\end{definition}
Note that the gauge transformations of $\sigma^*(g_{ab})$  can be identified (for instance by employing partial gauge condition \eqref{metric-gauge} along with $\sigma^*(D_{(a)}\lambda)=0$) with the action of diffeomorphisms and Weyl transformations whose parameters are associated to ghosts $\xi^a$ and $\lambda$.  Let us also note that if we allow $\Omega$ to vanish the sub-gPDE determined by $\Omega=0,\, D_{(a)}\Omega=0$ gives the gPDE reformulation of the conformal geometry, which is known in the literature in one or another version~\cite{Boulanger:2004eh,Joung:2021bhf,Dneprov:2022jyn}.

\subsection{Conformal-like on-shell GR}

From the field theory perspective the systems presented above are off-shell gauge theories, i.e. theories equivalent to a set of unconstrained  fields subject to gauge transformations. We are mostly interested in gravity-like theories, where fields are subject to nontrivial differential equations. The respective gPDE description can be obtained by considering a $Q$-subbundle of the initial off-shell system. In the case of off-shell GR \eqref{Q-riem} the $Q$-subbundle is defined as an infinite prolongation of the Einstein equations
\begin{align}\label{einst-eq}
        D_{(a)}(\tilde{R}_{bc}-\dfrac{\tilde{g}_{bc}}{d}\tilde{R})=0,     \quad \tilde{R}=\dfrac{2d}{d-2}\Lambda,
\end{align}
where $\tilde{R}_{bc}$, $\tilde{R}\equiv \tilde{g}^{bc}\tilde{R}_{bc}$ are local functions in $D_{(a)}\tilde{g}_{bc}$ corresponding to Ricci tensor and the scalar curvature respectively.  It is easy to see that $Q$ restricts to the submanifold and hence this indeed defines a gPDE, to which we refer in what follows as the \textit{on-shell GR}. 

In what follows we often encounter gPDEs defined as subbundles of other gPDEs i.e. jet-bundles. A convenient way to describe (coordinate) functions on such a subbundle is to regard them as the equivalence classes of functions modulo those vanishing on the subbundle. Alternatively, the restrictions of the ambient coordinates to the subbundle can be regarded as an overcomplete coordinate system therein.

Our aim now is to equivalently reformulate on-shell GR as a sub-gPDE of the conformal-like off-shell GR defined in~\bref{def-confgeom}. To this end consider a subbundle singled out by the following constraints:
\begin{align}\label{einst-AE}
D_{(a)}F_{bc}=0,\qquad \Omega\rho+\frac{g^{ab}}{2}D_a\Omega D_b\Omega=-\frac{\Lambda}{(d-1)(d-2)}\,,
\end{align} 
where
\begin{align}\label{Frho-def}
    F_{bc}\equiv D_b D_c \Omega- \Gamma_{bc}^{d}D_{d}\Omega+\Omega P_{bc} + \rho g_{bc},\\ \rho\equiv-\dfrac{1}{d}g^{bc}(D_{b}D_{c}\Omega- \Gamma_{bc}^{d}D_{d}\Omega+P_{bc}\Omega)
\end{align}
and $\Gamma_{bc}^{d}$, $P_{bc}$ are respectively the Christoffel symbols and the Schouten tensor seen as functions in the jets of the metric. 

Equations \eqref{einst-AE} are known as the \textit{almost Einstein equation}. However, usually they are interpreted as equations on $\Omega$ while metric $g_{ab}$ is considered fixed, see~\cite{BEG,Curry:2014yoa} for more details. Now  we treat~\eqref{einst-AE} as equations restricting both $g_{ab}$ and $\Omega$.
\begin{definition}\label{def-confeinst}
The sub-gPDE of the conformal-like off-shell GR \bref{def-confgeom}, which is determined by constraints \eqref{einst-AE}, is called \ROD conformal-like on-shell GR.
\end{definition}
The name is justified by the following:
\begin{prop}
     For $d\geq3$ conformal-like on-shell GR is equivalent to on-shell GR~\eqref{einst-eq}.
\end{prop}
\begin{proof}
First of all recall that $\Omega>0$. In terms of $\tilde g_{ab}=\Omega^{-2}g_{ab}$, the Einstein equations have the standard form~\eqref{einst-eq} while (the derivatives of) $\Omega$ and $\lambda$ form contractible pairs and can be eliminated. There remains to show that the Einstein equations rewritten in terms of $g_{ab}$ are equivalent to \eqref{einst-AE}. This can be checked using the well-known, see e.g.~\cite{BEG}, transformation rules of the Schouten tensor under Weyl transformations.
\end{proof}
It is important to stress that the above equivalence crucially relies on the condition $\Omega>0$. At the same time,  the conformal-like on-shell GR is perfectly well-defined without this condition and is going to be very instrumental in studying the boundary behaviour. Note also that the above system with no restrictions on $\Omega$ provides a gPDE description of tractor geometry though, to keep the exposition concise, we refrain from giving details here.

\subsection{Pre-minimal model for conformal-like on-shell GR}

In what follows we need a certain equivalent reduction (in the sense of definition \bref{equivred}) of conformal-like on-shell GR. This can be done in two steps. The first step is to concentrate on the sector of $g_{ab},\xi^a,\lambda$ and their jets. This sector is precisely the one that gives a gPDE description of conformal geometry (also regarded as the off-shell conformal gravity) so that one can eliminate contractible pairs as explained in \cite{Boulanger:2004eh}, see also \cite{Dneprov:2022jyn} for the discussion in similar language. Namely, 
    \begin{align}
    Q\Gamma_{ab}^c=\dots-D_a D_b\xi^c,\qquad
    QP_{ab}=\dots-D_a D_b \lambda
\end{align}
allows one to eliminate $\Gamma_{ab}^c$, $D_a D_b\xi^c$, $P_{ab}$, $D_a D_b \lambda$ as well as all their symmetrized total derivatives. This reduction is quite different for the cases $d=3$ and $d\geq4$, so in what follows we assume $d\geq 4$. However, generalization to $d=3$ is possible. The remaining jets of the metric have the meaning of the Weyl tensor $\We^{b}{}_{cde;(a)}$ and its covariant total derivatives. All the symmetries of the usual Weyl tensor are preserved: antisymmetry over pairs of indices, Bianchi identities, tracelessness, and so on. It is always possible to choose components of these tensors in such a way that they form a part of the coordinate system. However, for our purposes it is more convenient to keep all the components and use them as an overcomplete coordinate system. We introduce the notation $\Co_{dab}=-\dfrac{1}{(d-3)} \We^{c}{}_{dab;c}$.

Then the action of $Q$ on some coordinates in this sector is given by
 \begin{align}\label{gran-ish1}
    \begin{split}
        &Q g_{bc}=C_{b}{}^a g_{ac}+C_{c}{}^a g_{b a}+2\lambda g_{bc},\quad         Q\xi^b=\xi^a C_{a}{}^{b},\\
        &Q\lambda=\xi^a\lambda_a,\quad
        Q\lambda^b=C^{b}{}_a \lambda^a+\dfrac{1}{2}\xi^a \xi^d \Co^{b}{}_{ad},\\
        &Q C_{b}{}^c=C_{b}{}^a C_{a}{}^c+\lambda_b\xi^c-\lambda^c \xi_b+\delta_{b}^c \lambda_a \xi^a +\dfrac{1}{2}\xi^a\xi^d \We^{c}{}_{b a d}, 
    \end{split}
\end{align}
where the standard convention for raising and lowering indexes is used, for example $\lambda^{a} \equiv  g^{ab}\lambda_b$, and $C_{a}{}^{b} \equiv D_a\xi^b$. The  action of $Q$ on $\We^{b}{}_{cde;(a)}$ can also be obtained by a straightforward (but tedious) calculation.

In the second step we analyze the sector of $\Omega$.  If $\Omega_{(a)}$ denotes the restriction of $D_{(a)}\Omega$ to the system obtained in the 1st step, one finds that equations $D_{(a)}F_{bc}=0$ from \eqref{einst-AE} restricted to the surface defined by this reduction, takes the following form:
    \begin{align} \label{einst-svyazi-2}
    \begin{split}
        &\Omega_{ab}+g_{ab}\rho=0,\\
        &\Omega_{a_1\dots a_n}=0,\quad  n\geq3\,\\
        &\nabla_{a_1}\dots\nabla_{a_{n-3}}(\We^{d}{}_{a_n a_{n-2} a_{n-1}}\nabla_d \Omega + \Co_{a_n a_{n-2} a_{n-1}}\Omega)=0.
    \end{split}
\end{align}
where $\rho=-\frac{1}{d}g^{ab}\Omega_{ab}$ and for any degree $0$ coordinate $\varphi$, $\nabla_a\varphi$ is defined through $Q\varphi=\xi^a\nabla_a\varphi+\dots$. The first equation is just $F_{bc}=0$ restricted to the subbundle defined by previous reduction, the second is the totally-symmetric component of $D_{(a_1\ldots)}F_{a_{n-1}a_n}=0$, and the third one is the remaining irreducible component in $D_{(a_1\ldots)}F_{a_{n-1}a_n}=0$. These equations fix all the jets of $\Omega$ except for $\Omega$, $\Omega_a$ and $\rho$. For $n=3$ the third equation in \eqref{einst-svyazi-2} takes the form
\begin{align}
    \We^{d}{}_{a_n a_{n-2} a_{n-1}}\nabla_d \Omega + \Co_{a_n a_{n-2} a_{n-1}}\Omega=0
\end{align}
and is known as a part of the Fridrich equations, see e.g. \cite{Kroon:2016ink}. After taking into account the above equations and introducing $n^{a}\equiv g^{ab}\Omega_b$ in place of $\Omega_a$, the action of $Q$ on $\Omega_a, n^a,\rho$ takes the form:
\begin{align}\label{gran-ish2}
    \begin{split}
        &Q \Omega=\xi^a g_{ab}  n^b+\lambda\Omega,\quad Q n^b=-\xi^b\rho-n^a C_{a}{}^b+\lambda^b\Omega-\lambda n^b,\quad 
        Q\rho=-\lambda\rho-\lambda_a n^a.
    \end{split}
\end{align}
The results of this subsection can be summarized in the following proposition:
\begin{prop}
\label{prop:bulk-red}
    For $d\geq 4$ the gPDE defined in \bref{def-confeinst} is equivalent to its sub-gPDE $(E,Q,T[1]X)$ with the  following overcomplete set of  fiber coordinates $\{g_{bc}, \Omega, n^b,\rho, \xi^b, C_{b}{}^c, \lambda, \lambda^b, \We^{b}{}_{cde;(a)}, |a|\geq 0 \}$ which are understood modulo the ideal generated by the following constraints:
\begin{equation}
\begin{aligned}
\label{constraints-mb}
    &\Omega\rho +\frac{1}{2}g_{ab} n^a n^b=-\frac{\Lambda}{(d-1)(d-2)},\\
    &\nabla_{a_1}\cdots \nabla_{a_n}(\We^{b_3}{}_{cb_1 b_2}n^c-\Co^{b_3}{}_{b_1b_2}\Omega)=0, \quad n\geq 0.
\end{aligned}
\end{equation}
The action of $Q$ on all coordinates except curvatures $\We^{b}{}_{cde;(a)}$ is given by \eqref{gran-ish1}, \eqref{gran-ish2}.
\end{prop}


\section{Boundary systems and asymptotic symmetries}\label{sec:bound+asympt}

\subsection{Asymptotically simple GR as a gPDE with boundaries}
\label{sec:ggg}

Having obtained a description of gravity in the bulk as a gPDE one can immediately construct the induced gPDE $(i^{*}E,Q, T[1]\cJ)$ on the boundary $\cJ$.  
More specifically, we start with the gPDE defined in the Proposition~\bref{prop:bulk-red}, which encodes the conformal-like on-shell GR in the bulk. A slight but important modification is that fibers of $E$ over the boundary are extended by their own boundary by allowing $\Omega$ to take value $0$ (recall that in the bulk $\Omega >0$). 
In what follows we restrict ourselves to the local analysis and hence do not discuss global geometry of the space-time and its boundary. More specifically, we assume the boundary to have the topology of $S^{d-2}\times \fR$.

Now we identify a gPDE with boundaries which describes asymptotically simple GR. To this end we impose additional conditions which implements Penrose's definition of asymptotically simple spacetime in the gPDE terms. More specifically, we take $E_B \subset i^*E$ to be a sub gPDE of $i^*E$ determined by
\begin{align}\label{gran-usl}
    \Omega=0\,,\qquad Q\Omega=0\,. \qquad D_a\Omega \neq 0\,.
\end{align}
This gives a gPDE with boundaries $(E,Q,T[1]X,E_B,T[1]\cJ)$ which we refer to as \textit{asymptotically simple GR}. Here we keep using $Q$ to denote the homological vector field on $i^*E$ as well as on $E_B$ as these are restrictions of the initial $Q$ on $E$ to the respective submanifolds. Analogous systems for asymptotically-simple spacetimes are obtained by not imposing the Einstein equations. 

It is important to stress that if $D_a\Omega$ were nonvanishing  everywhere in $i^*E$, functions $\Omega$ and $Q\Omega$ would be independent on $i^*E$ so that setting $\Omega=0$ and $Q\Omega=0$ can be understood as an equivalent reduction. However, $D_a\Omega \neq 0$ is  imposed at $\Omega=0$ only so that it is better to regard~\eqref{gran-usl} as the boundary conditions determining asymptotically simple GR. In any case,~\eqref{gran-usl} effectively implements only minor restrictions on the moduli of solutions, which can be thought of as partial gauge conditions.  Another remark is that, as we discussed in Section~\bref{sec:gpde-wb}, the total space $E$ can be extended to a manifold with corners by allowing $\Omega \geq 0$ everywhere. From this perspective $E_B$ can be identified with the respective corner provided one also excludes points where $D_a\Omega \neq 0$. As an (overcomplete) coordinate system on $i^*E$ we use coordinates on $E$ restricted to $i^*E$ seen as a submanifold in $E$. In particular, $\Omega$ in \eqref{gran-usl} is, strictly speaking, a restriction of the initial coordinate $\Omega$ to $i^*E$.

Taking into account constraints~\eqref{gran-usl} in $(i^{*}E,Q,T[1]\cJ)$ results in the boundary gPDE $(E_B,Q,T[1]\cJ)$. The overcomplete set of fiber coordinates can be obtained by restricting the coordinates from Proposition~\bref{prop:bulk-red} to $E_B$ and is given by 
\begin{equation}   
\{g_{bc}, n^b,\rho, \xi^b, C_{b}{}^c, \lambda, \lambda^b,  \We^{b}{}_{cde;(a)}, |a|\geq 0 \}\,.
\end{equation}
The action of $Q$ on some of the coordinates is easily obtained by restricting~\eqref{gran-ish1}, \eqref{gran-ish2}:
\begin{align}
    \begin{split}
        &Q g_{bc}=C_{b}{}^a g_{ac}+C_{c}{}^a g_{b a}+2\lambda g_{bc},\quad Q\xi^b=\xi^a C_{a}{}^{b},\\
        &Q n^b=-\xi^b\rho-n^a C_{a}{}^b-\lambda n^b,\quad Q\lambda^b=C^{b}{}_{a} \lambda^a+\dfrac{1}{2}\xi^a \xi^c \Co^{b}{}_{ac},\\
        &Q\rho=-\lambda\rho-\lambda_a n^a,\quad    Q\lambda=\xi^a \lambda_a,\\
        &Q C_{b}{}^c=C_{b}{}^a C_{a}{}^c+\lambda_b\xi^c-\lambda^c \xi_b+\delta_{b}^c \lambda_a \xi^a +\dfrac{1}{2}\xi^a\xi^d \We^{c}{}_{b a d}\,.
    \end{split}
\end{align}
At the same time constraints \eqref{constraints-mb} take the form:
\begin{align}\label{gran-svyazi}
    \begin{split}       
        &g_{ab} n^a n^b=-\frac{2}{(d-1)(d-2)}\Lambda,\\
        &\left(\nabla_{a_1}\cdots \nabla_{a_n}(\We^{b_3}{}_{cb_1 b_2}n^c-\Co^{b_3}{}_{b_1b_2}\Omega)\right)\big|_{\Omega=0}=0, \quad n\geq 0,\,,\\
                 \end{split}
\end{align}
and, finally, the last constraint is given by $\xi^a g_{ab}n^b=0$ and originates from $Q\Omega=0$.
\begin{definition}\label{def:gran-opr}
The above gPDE $(E_B,Q,T[1]\cJ)$ is refereed to as the boundary gPDE  for asymptotically simple GR.
\end{definition}

\ruth{} As we will see later it is very convenient to work in terms of the minimal model of the system~\bref{def:gran-opr}. However, the respective minimal gPDE crucially depends on the value of the cosmological constant. As we are mostly interested in the null-infinity we assume $\Lambda=0$ unless otherwise specified.

\subsection{Minimal model for the boundary gPDE of asymptotically simple GR}
\label{sec:min-model-boundary}

We now take $\Lambda=0$ and find a minimal model of the boundary gPDE  obtained in the previous Section. We have the following:
\begin{prop}\label{prop:gran-utv-Lambda0}
    In the case $\Lambda=0$ and $g_{ab}$ of Lorentz signature,  gPDE $(E_B,Q,T[1]\cJ)$ defined in \bref{def:gran-opr} is equivalent to its subbundle determined by the following conditions:
      \begin{align}
&g_{ab} =
\begin{pmatrix}
0 & 1 & 0\\
1 & 0 & 0\\
0 & 0 & -\delta_{AB}
\end{pmatrix},\quad
n^a=
\begin{pmatrix}
0 \\
1 \\
0 
\end{pmatrix},
\\[10pt]
&C_{a}{}^b =
\begin{pmatrix}
-\lambda & 0 & C_A\\
0 & -\lambda & 0\\
0 & -C_A & \rho_{A}{}^B-\lambda \delta_{A}^B
\end{pmatrix},\\[10pt]
&\rho=0,\quad \xi^\Omega=0,\lambda^\Omega=0,
\end{align}
where we used the adapted partition of indexes $\{a\}=\{\Omega,u,A\}$, $A=1,\ldots,d-2$ and introduced the following new coordinates: $\rho_{AB}\equiv C_{[AB]}$, $C_A\equiv C_{\Omega\, A}$. Among the constraints~\eqref{gran-svyazi} on the degree-zero variables there only remain:
\begin{align}\label{Weyl-lambda0}
    \nabla_{a_1}\cdots\nabla_{a_n}\We_{b_3 u b_1 b_2}-\sum_{i=1}^{n}g_{u a_i}\nabla_{a_1}\cdots\hat{\nabla}_{a_i}\cdots \nabla_{a_n}\Co_{b_3 b_1 b_2}=0, \quad n\geq0.
\end{align}
Note that $g_{u a_i}=\delta_{a_i\Omega}$ on the subbundle. 
\end{prop}
We denote the minimal model introduced in the above Proposition by $(E_B^{\text{min}},Q,T[1]\cJ)$. This gPDE is explicitly defined as a sub-gPDE of $(E_B,Q,T[1]\cJ)$ which, in its turn, is a sub-gPDE of $i^*E$.
\begin{proof}
As usual, the proof is based on the identification of contractible pairs. Using 
\begin{align}
    Q n^b=-\xi^b \rho -n^a C_{a}{}^b-\lambda n^b
\end{align}
and taking into account $n^b\neq 0$ one can set:
\begin{align}
    \begin{split}
        n^\Omega=0,\quad C_{u}{}^\Omega=-\xi^\Omega\rho,\qquad
        n^u=1,\quad C_{u}{}^u=-\xi^u\rho -\lambda,\qquad
        n^A=0,\quad C_{u}{}^A=-\xi^A\rho\,,
    \end{split}
\end{align}
which also gives $g_{uu}=0$ thanks to the first constraint in~\eqref{gran-svyazi}. Using then 
\begin{align}
    Qg_{ua}=-\xi^\Omega\rho g_{\Omega a}-\xi^u\rho g_{ua}-\xi^B\rho g_{Ba}+C_{a}{}^b g_{ub}+\lambda g_{ua}\,,
\end{align}
we can eliminate  $g_{ua}$ as well as $C_{a}{}^b g_{ub}$. Note that $g_{ub}\neq0$ because of $det(g_{ab})\neq0$. More precisely, we set 
\begin{align}
    \begin{split}
        &g_{u\Omega}=1,\quad C_{\Omega}{}^{\Omega}=\xi^u\rho - \lambda,\\
        &g_{uA}=0,\quad C_{A}{}^{\Omega}=\xi^\Omega\rho g_{\Omega A}+\xi^B\rho g_{BA}.
    \end{split}
\end{align}
The second constraint in \eqref{gran-svyazi} then gives $\xi^\Omega=0$. Using 
\begin{align}
    Q\rho=-\lambda\rho-\lambda^\Omega,
\end{align}
allows us to set $\rho=0$ and $\lambda^\Omega=0$. Furthermore, using 
\begin{align}
    Qg_{\Omega\Omega}=2C_{\Omega}{}^u+2C_{\Omega}{}^B g_{B\Omega}+2\lambda g_{\Omega\Omega}
\end{align}
we can set $g_{\Omega\Omega}=0$ and $C_{\Omega}{}^u=-C_{\Omega}{}^B g_{B\Omega}$. Similarly,
\begin{align}
    Qg_{\Omega A}=\lambda g_{\Omega A}+ C_{\Omega}{}^B g_{B A}+C_{A}{}^u
\end{align}
allows us to set $g_{\Omega A}=0$ and $C_{A}{}^u=-C_{\Omega}{}^B g_{BA}$. Finally,
eliminating the remaining components $g_{AB}$ of the metric we set
\begin{align}
    g_{AB}=-\delta_{AB}, \quad C_{(AB)}=\lambda \delta_{AB}.
\end{align}
\end{proof}

To summarize, we have explicitly found a minimal model of the boundary gPDE for asymptotically simple GR. Its overcomplete fiber coordinates are $\{\xi^u,C^A,\xi^A, \rho_{A}{}^B,\lambda,\lambda^u,\lambda^A,  \We^{b}{}_{cde;(a)},|a|\geq0 \}$. The action of $Q$ on the degree $1$-coordinates $\xi^A,\lambda,\lambda_A,\rho_A{}^B$ is given by:
\begin{align}\label{asymptotic ghosts1}
    \begin{split}
            &Q\xi^A=\xi^B\rho_{B}{}^A-\xi^A\lambda,\\
            &Q\rho_{A}{}^{B}=\rho_{A}{}^C \rho_{C}{}^B +\lambda_A\xi^B-\lambda^B \xi_A +\frac{1}{2}\xi^C \xi^D \We^B{}_{A CD},\\
            &Q\lambda=\xi^A\lambda_A,\\
            &Q\lambda^A=\rho^{A}{}_{B}\lambda^{B}-\lambda\lambda^{A}+\dfrac{1}{2}\xi^C\xi^D \Co^{A}{}_{CD}+\xi^{u}\xi^{D}\Co^{A}{}_{uD}.
    \end{split}
\end{align}
Setting the curvatures (these enter the right hand sides multiplied by $\xi^a\xi^b$) to zero gives the Chevalley-Eilenberg differential of the $so(d-1,1)$ subalgebra of $iso(d-1,1)$ algebra. This subalgebra can be identified with the conformal algebra of a $d-2$-dimensional flat space. Moreover, setting to zero only the components $\Co^{A}{}_{uD}$ gives the respective sector of the minimal model of the conformal geometry in $d-2$ dimensions. Note, however, that the entire system differs form that of conformal geometry. In particular, extra curvatures are present and the action of $Q$ on the curvatures is different.

The action of $Q$ on the remaining degree-1 coordinates reads as:
\begin{align}\label{asymptotic ghosts2}
    \begin{split}
        &Q\xi^u=-\xi^u\lambda-\xi^A C_A,\\
        &Q C^A=C^B \rho_{B}{}^A+\lambda^u \xi^A-\lambda^A\xi^u+\dfrac{1}{2}\xi^C\xi^D \We^A{}_{\Omega CD},\\
        &Q\lambda^u=C^A \lambda_A-\lambda \lambda^u +\dfrac{1}{2}\xi^C \xi^D \Co_{\Omega CD}+\xi^{u}\xi^{D} \Co_{\Omega u D}.
    \end{split}
\end{align}
With all the curvatures set to zero, the actions of $Q$ is that of the Chevalley-Eilenberg differential of $iso(d-1,1)$, where \eqref{asymptotic ghosts1} corresponds to $so(d-1,1)$ while \eqref{asymptotic ghosts2} to the $iso(d-1,1)$ translations.

Fields parameterizing solutions of the above minimal model are the $iso(d-1,1)$ connection on the boundary along with the bunch of the degree zero fields (curvatures) some of which are expressed in term of the connection through the equations of motion (and generally impose some differential equations on the connection) while the remaining ones are independent fields. It is of course natural that the boundary system can be formulated in terms of the Poincar\`e connection because the minimal model of the bulk gravity has an analogous formulation. Similar, but not identical formulations of the asymptotically simple GR were considered in~\cite{Nguyen:2020hot},\cite{Herfray:2021qmp}. More details on the field theory encoded in the above minimal model are given in Section~\bref{sec:eom-min}.

In the next two sections, in order to agree with the standard conventions for connections and curvatures, we redefine all fiber coordinates $\varphi$ such that $\gh\varphi=1$ as $\varphi\rightarrow -\varphi$. In particular, this affects the explicit  formulas for the action of $Q$ on fiber coordinates (an alternative way is to reverse the sign at the vertical part of $Q$).

\subsection{Boundary conditions and BMS symmetries}\label{sec:Bound+BMS}

Now we plan to identify a proper counterpart of the BMS boundary conditions in this setup. Strictly speaking the minimal model constructed in the previous section is too ``minimal'' to incorporate a sub-gPDE of boundary conditions as a regular submanifold.  Nevertheless it is not difficult to identify, generally non-regular, constraints which do the job, giving a rather concise description of the boundary conditions and asymptotic symmetries.

As explained in Section~\bref{sec:as-sym-gPDE} in this setup we are forced to allow $Y$ to depend on jets of sections of $E^{min}_B$ (recall that we treat the gPDE of boundary conditions as a subbundle in $E^{min}_B$). Here we use $D^\theta_{\mu}$ to denote the total derivative in $\theta^\mu$ direction (it's a total derivative in the super-jet bundle of $E^{min}_B$ and should not be confused with the total derivative in the initial jet-bundle from which $E^{min}_B$ has been constructed). For instance, if $\sigma$ is a section and $\sigma^*(\lambda_A)=\lambda_{A\mu}(x)\theta^\mu$ then $\sigma^*(D^\theta_{\mu} \lambda_A)=\lambda_{A\mu}(x)$. Note that $\sigma^*(D^\theta_{\mu}D^\theta_{\nu} \lambda_A)=0$ by the degree reasoning.

We define $E_\cJ\subset E_B^{min}$ as a zero locus of the constraints defined on $E^{min}_B$. We first introduce constraints which set the frame field encoded in $\xi^a$ to be a fixed frame: 
\begin{equation}
\label{constr-1}
\xi^A-e^A\sim 0, \qquad \xi^u-\theta^u \sim 0\,,
\end{equation}
where we use adapted coordinates $y^\alpha$ and $u$ on the boundary $\cJ$ and assumed for simplicity that $e^A=e^{A}{}_\alpha(u,y) \theta^\alpha$. 
It is easy to see that $Q$ is not tangent to the surface and hence extra boundary conditions are necessary. Consider the following extra constraints:
\begin{equation}
\label{constr}
\lambda \sim 0\,, \quad \lambda_A\xi^A \sim 0\,, \quad C_A\xi^A \sim 0\,, \quad \dJ e^A+\xi^B\rho_B{}^A  \sim 0\,,    
\end{equation}
where the last three ones coincide with $Q\lambda$, $Q(\xi^u-\theta^u)$, and $Q(\xi^A-e^A)$ modulo terms proportional to $\lambda$.  This can be easily seen using \eqref{asymptotic ghosts1} and \eqref{asymptotic ghosts2} as well as the following representation:
\begin{equation}
C_A\xi^A=Q\xi^u-\xi^u\lambda\,, \qquad 
\rho^A{}_B\xi^B=Q\xi^A-\xi^A\lambda\,.
\end{equation}

Constraints~\eqref{constr} and \eqref{constr-1} define an ideal $\cI_\cJ$ in the algebra of functions on $E_B^{min}$ introduced in Proposition~\bref{prop:gran-utv-Lambda0}. It is easy to see that $Q$ is well defined on the quotient as $\cI_\cJ$ is $Q$-invariant. Because some of the constraints are quadratic, the quotient is not an algebra of functions on a regular subbundle. However, we can still think of it as determining a $Q$-subbundle $E_\cJ$ which is defined in the algebraic sense only.  This does not really lead to problems  because its prolongation to jets of supersections is a regular submanifold, provided we restrict ourselves to sections such that the corresponding frame field is invertible. In this sense working in terms of $E_\cJ$ only gives an economical framework to analyse asymptotic symmetries. All the steps can be repeated in terms of its jet-prolongation which is a genuine subbundle of the jet-bundle. Disregarding the above subtlety, constraints~\eqref{constr} and \eqref{constr-1} define a gauge PDE with boundary in the sense of Definition~\bref{def:BgPDE}. Indeed, $E_\cJ$ is a sub-gPDE of the $E_B^{min}$ which, in turn, is defined as a sub-gPDE of $i^*E$.

Now we are ready to study gauge symmetries that preserve the gPDE of boundary conditions. Consider a gauge parameter vector field
\begin{equation}
\label{Yparam}
Y=\epsilon^u\dl{\xi^u}+\epsilon^A\dl{\xi^A}+\bar \lambda \dl{\lambda}+ \bar \lambda_A \dl{\lambda_A}+\bar\rho^{AB}\dl{\rho^{AB}}+\bar C^{A}\dl{C^{A}}+\bar\lambda{}^u\dl{\lambda^u}\,,
\end{equation}
where $\epsilon^u,\epsilon^A,\bar\lambda, \ldots$  are functions in $x$ while $\bar\lambda_A,\bar C^A,\bar\rho^{AB}$ are also allowed to depend on the $\theta$-jets of $\lambda_A,C^A,\rho^{AB}$. Note that the component $\bar\lambda{}^u\dl{\lambda^u}$ clearly preserves the constraints and hence correspond to trivial asymptotic symmetries. 

We are interested in $Y$ such that the respective symmetry transformations preserves the ideal and hence induces a symmetry transformation that take solutions of $E_\cJ$ to solutions. We have:
\begin{equation}
\label{cond-asym}
    \dJ \sigma^*(Yf)+\sigma^*(YQf)=0 \quad \forall f \in \cI
\end{equation}
This should hold for all section of the gPDE of boundary condition, i.e. sections of $E_B^{min}$ such that $\sigma^*(\text{``constraints''})=0$. 

Taking $f=\lambda$ gives 
\begin{equation}
\label{cons-lambda}
\sigma^*(\dJ \bar\lambda-\epsilon^A \lambda_A+\xi^A \bar \lambda_A)=0
\end{equation}
This implies $\d_u \bar\lambda=0$ because we assumed $\theta^u$ unconstrained and because $\sigma^*(\lambda_A)=\lambda_{AB}(x) e^B$ for some $\lambda_{AB}(x)=\lambda_{BA}(x)$ thanks to $\sigma^*(\xi^A\lambda_A)=0$. Furthermore, \eqref{cons-lambda} also implies:
\begin{equation}
   e^\alpha{}_{B}\d_\alpha \bar\lambda - \epsilon^A \lambda_{AB}+\sigma^*(\bar\lambda_B) =0\,.
\end{equation}
This can be solved for $\bar\lambda_B$ by e.g. $\bar\lambda_B=-e^\alpha{}_{B}\d_\alpha \bar\lambda +e^\alpha{}_{B} \epsilon^A  D^{\theta}_\alpha \lambda_{A}$. Indeed, $\sigma^*(D^{\theta}_\alpha \lambda_{A})=\lambda_{AB}(x)e^{B}{}_\alpha$ and hence imposes no restrictions on $\d_\alpha \bar\lambda$.

Taking $f=\xi^u-\theta^u$ in \eqref{cond-asym} one finds
\begin{equation}
\label{barcdef}
\dJ \epsilon^u+\sigma^*(\epsilon^u \lambda -\xi^u \bar\lambda+\epsilon^A C_A -\xi^A \bar C_A)=0
\end{equation}
This implies $\epsilon^u=u\bar\lambda(y) +T(y)$, with $T$ unconstrained. Moreover, the remaining equation can be satisfied by taking $\bar C_A=e^\alpha{}_{A}(\d_\alpha \epsilon^u + \epsilon^B D^\theta_\alpha C_B)$.

Taking $f=\xi^A-e^A$ one gets
\begin{equation}
\label{xi-check}
\dJ \epsilon^A+\sigma^*(-\epsilon^B\rho_{B}{}^{A}+\xi^B\bar\rho_{B}{}^{A}+\epsilon^{B}\lambda-\xi^B\bar\lambda)=0
\end{equation}
Thanks to $\sigma^*(\dJ e^A+\rho^A{}_B\xi^B)=0$
one finds that $\sigma^*(\rho^A{}_B)=\omega^{A}{}_{B\mu }\theta^\mu$, where $\omega^{A}{}_{B\alpha}(u,y)$ can be expressed in terms of $e^A=e^{A}{}_{\alpha}(u,y)\theta^\alpha$ through standard formulas for Levi-Civita connection and $\omega_{B}{}^{A}{}_{u}=\sigma^{*}(e^{\alpha}{}_{B}\d_{u}e^{A}{}_{\alpha})$. Then, in terms of $\epsilon^{\alpha}\equiv e^{\alpha}{}_{A}\epsilon^{A}$, the equation ~\eqref{xi-check} implies
\begin{equation}\label{xi-check2}
\d_u\epsilon^{\alpha}=0,\quad \d_\alpha\epsilon_\beta-\Gamma^{\gamma}_{\alpha\beta}(e)\epsilon_\gamma+e^{A}{}_{\alpha} e^{B}{}_{\beta}\sigma^*(\bar\rho_{AB})-g_{\alpha \beta}\bar\lambda=0.
\end{equation}
Here, similarly to the standard formulas, $g_{\alpha\beta}\equiv e_{A\alpha}e^{A}{}_{\beta}$ and $\Gamma^{\gamma}{}_{\alpha\beta}\equiv e^{\gamma}{}_{A}(\d_\alpha e^{A}{}_{\beta}-e^{B}{}_{\beta}\omega_{B}{}^{A}{}_\alpha)$. The antisymmetric part of the second equation in \eqref{xi-check2} can be solved for $\bar\rho_{AB}$, and the symmetric part is nothing but a conformal Killing equation. In particular this fixes $\bar\lambda$ in terms of $\epsilon^A$.

Finally,  there remains to check~\eqref{cond-asym} for the last three constraints from~\eqref{constr}.
However, these three are all of the form $Qg$, modulo terms proportional to $\lambda$, with $g$ being $\lambda$ or $\xi^u-\theta^u$ or $\xi^A-e^A$. It follows \eqref{cond-asym} always holds because
\begin{equation}
    \dJ \sigma^*(YQg)+\sigma^*(YQQg)=\dJ \sigma^*(YQg)=-\dJ (\dJ \sigma^*(Yg))=0\,,
\end{equation}
where in the last equality we made use of \eqref{cond-asym}, with $f$ replaced by $g$, and the fact that $Y$ was chosen in such a way that \eqref{cond-asym} holds for $f$ being $\lambda$ or $\xi^u-\theta^u$ or $\xi^A-e^A$.

In this way we are left with $Y$ parameterized by $u$-independent $T$ and $\epsilon^\alpha$. Interpreting $\epsilon^u(u,y),\epsilon^\alpha(y)$ as components of a vector field on the boundary it is easy to check that this is precisely BMS vector field on the boundary, which encodes conformal isometries of $d-2$-dimensional space and supertranslations. More specifically, the BMS vector field on $E_\cJ$ reads as
\begin{equation}
    \epsilon^{BMS}=(u\bar\lambda+ T(y))\dl{u}+\epsilon^{\alpha}(y)\dl{y^\alpha}\,,
\end{equation}
where we use adapted coordinates $u,y^\alpha$ on $\cJ$ and where $T(y)$ is a generic function in $y^\alpha$, $\epsilon^\alpha(y)$ are components of a conformal Killing vector in $d-2$ dimensions, and $\bar\lambda$ is determined by~\eqref{xi-check2}. This is precisely how the infinitesimal BMS transformations act as symmetries of the conformal Carrollian geometry, see e.g.~\cite{Duval:2014uva} for more details.

To make sure we are dealing with nontrivial asymptotic symmetries one should, strictly speaking, show that these symmetries are not equivalent to trivial. I.e. that $\commut{Q}{Y}|_{E_\cJ}$ can not be represented as $\commut{Q|_{E_\cJ}}{Y^\prime}$ for some vertical vector field $Y^\prime$ on ${E_\cJ}$. Considering $Y^\prime$ as a representative of an equivalence class of vertical vector fields tangent to $E_\cJ$ modulo those vanishing on $E_\cJ$ one can assume that $Y^\prime \lambda=Y^\prime \xi^A=Y^\prime \xi^u=0$. Repeating the analysis of this section for such $Y^\prime$ one concludes that $\sigma^*(\bar \lambda_A)=\sigma^*(\bar C_A)=\sigma^*(\bar \rho^A_B)=0$.  Considering, for instance, $\sigma^*(\commut{Q}{Y^\prime} C_A)$ one finds that $(\delta_{Y^\prime}\sigma)^*(C_A)= \sigma^*(e^A\bar\lambda^u)$. Then introducing components $C_{AB}$ as $\sigma^*(C_{A})=e^B C_{BA}(x)$,  the transformation takes the form:
\begin{equation}
\delta_{Y^\prime} C_{AB}=\eta_{AB}\bar\lambda{}^u
\end{equation}
so that it cannot affect the trace-free components of coordinates $(C_{AB})$.  As will be shown in the next section, these components parameterize the asymptotic shear. At the same time transformations with nontrivial $\epsilon^u,\epsilon^\alpha$ do affect the asymptotic shear. 

\subsection{Field-theoretical interpretation of the minimal model}
\label{sec:eom-min}

As we have seen the minimal model $(E^{min}_B,Q,T[1]\cJ)$ of the boundary gPDE for asymptotically simple GR, defined in Section~\bref{sec:min-model-boundary}, plays a crucial role in our approach to asymptotic symmetries. In this section we study its solutions and gauge symmetries and explain how the BMS symmetries can be derived in these terms.

We now study the space of solutions, i.e. sections of $(E^{min}_B,Q,T[1]\cJ)$ satisfying $\dJ \circ\sigma^{*}=\sigma^{*}\circ Q$. By some abuse of notation we introduce the following parameterization of sections: 
\begin{align}
        \sigma^{*} \rho_{A}{}^{B}=\omega_{A}{}^{B},\quad \sigma^{*}\xi^A=e^{A},\quad
        \sigma^{*}\xi^u=l\,,
\end{align}
where all the new functions are linear in $\theta^\mu$, i.e. can be seen as 1-forms on $X$, by the degree reasoning.  For the remaining fiber coordinates we take $\sigma^{*}\phi=\phi(x,\theta)$. The equations of motion in the sector of degree-1 fiber coordinates read as:
\begin{equation}\label{asimpt-ghostsform}
    \begin{gathered}
         \dJ e^A+\omega^{A}{}_B e^B+\lambda e^A=0,\qquad \dJ\lambda+e^A\lambda_A=0,\qquad \dJ l+\lambda l-e^A C_A=0\,,\\
                     \dJ\omega_{A}{}^{B}+\omega_{A}{}^C \omega_{C}{}^B +\lambda_Ae^B-\lambda^B e_A =\frac{1}{2}e^C e^D \We_{A}{}^{B}{}_{CD},\\            
            \dJ\lambda^A+\omega^{A}{}_{B}\lambda^{B}-\lambda\lambda^{A}=-le^D \Co^{A}{}_{uD}-\dfrac{1}{2}e^C e^D \Co^{A}{}_{CD},\\
        \dJ C^A+ \omega^{A}{}_B C^B+\lambda^{u} e^A-\lambda^A l=\dfrac{1}{2}e^Ce^D \We_{\Omega}{}^{A}{}_{CD},\\
        \dJ\lambda^{u}+C_A\lambda^A-\lambda \lambda^{u} =-l e^{D} \Co_{\Omega uD}-\dfrac{1}{2}e^C e^D \Co_{\Omega CD}\,.
    \end{gathered}
\end{equation}
These are Cartan structure equations for the $iso(1,d-1)$ connection written in the special basis, where the $so(1,d-1)$ subalgebra is made explicit as the conformal algebra in $d-2$-dimensions and in contrast to the usual Cartan description of Riemannian or Einstein geometry these equations are defined in $d-1$-dimensional space rather than $d$-dimensional one. Moreover, the curvatures appearing in the right hand sides of the above equations are subject to specific constraints. For instance, components of the curvature in the sector of varibales $e^A,l$ and $\lambda$ vanish. 
    
 In the case of $d=4$ the curvature of this connection contains 5 independent components, namely  $\Co_{\Omega cd}$ and $\Co_{Bcd}$ (other components vanish in $d=4$), which can be identified with Newman-Penrose coefficients $\Psi_{4}, \Psi_{3}, Im\Psi_{2}$ encoding the gravitational radiation. At the same time, fields $\We_{A \Omega \Omega B}$, $\We_{\Omega u A \Omega}$, and $\We_{\Omega u u \Omega}$
 also contain 5 independent components which correspond to the remaining Newman-Penrose coefficients $\Psi_0, \Psi_1$ and $Re(\Psi_2)$, see e.g.~\cite{Newman:1961qr}. Let us stress that in contrast to the former, the latter 5 components do not enter the Cartan structure equations and hence can not be interpreted as components of the curvature of the $iso(1,d-1)$-connection on the boundary. These are known to capture the longitudinal information and indeed are not described by the curvature~\cite{Ashtekar:1981hw}, see also~\cite{Herfray:2020rvq,Herfray:2021qmp} for more details.\footnote{They are analogous to the components of subleading modes appearing in the near-boundary analysis of critical fields in the AdS/CFT context, see e.g.~\cite{Fefferman:2007rka,Skenderis:2002wp}. For generic fields, these modes were described in~\cite{Bekaert:2012vt,Bekaert:2013zya} within a version of gPDE approach.}

Let us introduce the components of the dual frame according to 
$e^A=\sigma^*(\xi^A)=e^{A}{}_{\mu}\theta^\mu$ and $l=\sigma^*(\xi^u)=l_{\mu}\theta^\mu$ and restrict to sections with invertible frame. Components $(e^{\mu}{}_{A},n^\mu)$ of the frame are introduced via
\begin{align}
    n^\mu l_\mu=1, \quad n^{\mu}e^{A}{}_{\mu}=0,\quad l_{\mu}e^{\mu}{}_{A}=0,\quad e^{A}{}_\mu e^{\mu}{}_{B}=\delta^{A}_{B}.
\end{align}
Taking into account the constraints on the curvature one can check that as independent components of the connection one can take $\{l_\mu, e^{A}{}_{\mu},e^{\nu}{}_{A}\lambda_\nu, C_{(AB)}\}$ because the remaining components can be expressed through them.

The parameterization of the space of solutions to \eqref{asimpt-ghostsform} can be described more efficiently if one makes use of the gauge freedom \eqref{predv-gaugetransf}. Introducing gauge parameter vector field $Y$ as in~\eqref{Yparam} and assuming coefficients to depend on $x$ only the gauge transformation for $\lambda_\mu$ reads as
\begin{align}\label{g-fix-lambdaold2}
    \lambda_\mu\rightarrow\lambda_\mu +\partial_\mu\bar{\lambda}-\epsilon^{A}\lambda_{A\mu}+e^{A}{}_{\mu}\bar{\lambda}_A\,,
\end{align}
so that the following gauge condition can be imposed:
\begin{align}\label{gauge-fix-1}
    e^{\nu}{}_{A}\lambda_\nu=0\,.
\end{align}
In this gauge the components of the gauge parameter satisfy $\bar{\lambda}_A=-e^{\mu}{}_{A}(\partial_\mu\bar{\lambda}-\epsilon^B\lambda_{B\mu})$. In a similar way, we can achieve $C^{A}{}_{A}=0$, leading to further relations between gauge parameters:
\begin{align}
      \bar{\lambda}^{u}=\dfrac{1}{d-2}e^{\mu}{}_{A}(\partial_\mu \bar{C}^A-\bar{C}^B\omega_{B}{}^{A}{}_{\mu}+ C^{B}{}_{\mu}\bar\rho_{B}{}^{A}+\lambda^{u}{}_{\mu}\epsilon^A-\lambda^{A}{}_{\mu}\epsilon^{u}).
 \end{align}
Furthermore, using
\begin{align}
            \delta l_\mu=\partial_\mu \epsilon^{u}+\epsilon^u \lambda_\mu-l\bar{\lambda}+\epsilon^A C_{A\mu}-e^{A}{}_{\mu}\bar{C}_A\,,
\end{align}
the following gauge can be reached $l_\mu=\partial_\mu u$, where $u$ is a function of $x^\mu$ satisfying $n^{\mu}\partial_\mu u=1$. Function $u$ is often employed in the literature on BMS symmetries and it is convenient to take it as one of the coordinate functions $\{x^\mu\}\rightarrow \{u,y^{\alpha}\}$, $\alpha=1,\ldots,d-2$. Let us also list the constraints on gauge parameters, which ensure preservation of $l_\mu=\partial_\mu u$:
\begin{align}
    \bar{C}_{B}=e^{\mu}{}_{B}(\partial_\mu\epsilon^{u}+\epsilon^{u}\lambda_{\mu}+\epsilon^{A}C_{A\mu}),
    \qquad
\bar{\lambda}=n^\mu(\partial_\mu\epsilon^{u}+\epsilon^{u}\lambda_\mu+\epsilon^{A}C_{A\mu})\,.
\end{align}

To summarize: by imposing gauge condition as explained above one can parameterize the connection in terms of algebraically independent components $\{e^{A}{}_{\mu}, C_{(AB)}|_{tf}\}$ (of course there can be nontrivial differential constraints following from the constraints on the curvature). In so doing ${e^{A}}{}_{\mu}$ encodes the degenerate metric $g_{\mu\nu}\equiv e^{A}{}_{\mu}g_{AB} e^{B}{}_{\nu}$ whose kernel is generated by $n=\frac{\partial}{\partial u}$, while $-\frac{1}{2}C_{(AB)|_{tf}}$ is the so-called asymptotic shear, see e.g.\cite{Ashtekar:2014zsa}, which  parameterize torsion-free and metric-compatible affine connections on the boundary. Recall that such a connection is not unique if metric is degenerate. The geometry determined by $g_{\mu\nu}$ and $n^\mu$ defined up to an overall Weyl-like rescalings is often refereed to as conformal Carroll geometry.

The setup of this section gives an alternative framework to study asymptotic symmetries, which in contrast to the more algebraic approach of Section~\bref{sec:Bound+BMS}, is somewhat analogous to the standard analysis, see e.g.~\cite{Nguyen:2020hot}, \ruth{} where the first-order formalism is also employed. Let us sketch how asymptotic symmetries  can be found in this framework. First of all one imposes boundary conditions on sections of $E_B^{min}$ and then, in order to simplify the system, one imposes partial gauge conditions, e.g. the one discussed above. In the next step one studies gauge transformations that preserve this boundary conditions. For instance, to arrive at BMS symmetries in the present framework it is enough to fix a concrete frame $e^A=e^{A}{}_\mu(x)\theta^\mu$
and set $\lambda=0$. BMS symmetries are then obtained as the residual symmetries preserving these boundary conditions. Note that these boundary conditions correspond to only a subset of the conditions~\eqref{constr-1} and~\eqref{constr} of Section~\bref{sec:Bound+BMS}. The remaining conditions correspond to solving some of the equations of motion and imposing partial gauge conditions, cf. Remark~\bref{rem:reduction}.


\subsection{Asymptotically (A)dS spaces}
In the above analysis we concentrated on asymptotically flat spacetimes. It turns out that the boundary system $(E_B,Q,T[1]\cJ)$ defined in~\bref{def:gran-opr} also works in the case of asymptotically (A)dS spacetimes\footnote{In our analysis $\Lambda>0$ for asymptotically AdS spacetimes or $\Lambda<0$ for asymptotically dS spacetimes, since we work in the signature $(+,-,\dots, -)$. See e.g.~\cite{Kroon:2016ink}}. In this case the metric induced on the boundary is nondegenerate and hence the respective minimal model differs substantially from the case of $\Lambda=0$. More precisely, we have
\begin{prop}
    For $\Lambda\neq0$ Gauge PDE \eqref{def:gran-opr} is equivalent to its sub-gPDE defined as follows:
    \begin{align}
    \begin{split}
    &g_{ab} =
\begin{pmatrix}
-\tilde{\Lambda} & 0 \\
0 & \eta^{\varepsilon}_{AB} 
\end{pmatrix},\quad
n^a=
\begin{pmatrix}
1 \\
0  \\
\end{pmatrix},\quad
C_{a}{}^{b}=
\begin{pmatrix}
-\lambda & 0 \\
0 & \rho_{A}{}^{B}-\lambda\delta_{A}^{B},
\end{pmatrix},\\
&\rho=0,\quad \xi^\Omega=0,\quad \lambda^\Omega=0,
    \end{split}
    \end{align}
where the adapted partition of indexes, e.g. $\{a\}=\{\Omega,A\}$, $A=0,\ldots,d-1$ has been employed and 
\begin{align}
    \eta^{\varepsilon}_{AB}\equiv (\varepsilon,-1,\dots,-1),\quad \rho_{AB}\equiv C_{[AB]},\quad \tilde{\Lambda}\equiv\frac{2}{(d-1)(d-2)}\Lambda,\quad \varepsilon\equiv\sign \Lambda\,.    
\end{align}
Among the constraints on the degree-zero variables
\eqref{gran-svyazi} there only remain:\begin{align}\label{gran-Weyl-lambda}
    \nabla_{a_1}\cdots\nabla_{a_n}\We_{b_3\Omega b_1 b_2}-\sum_{i=1}^{n}g_{\Omega a_i}\nabla_{a_1}\cdots\hat{\nabla}_{a_i}\cdots \nabla_{a_n}\Co_{b_3 b_1 b_2}=0, \quad n\geq0,
\end{align}
where the hatted symbols are assumed omitted and $g_{\Omega a_i}=-\tilde{\Lambda}\delta_{\Omega a_i}$.
\end{prop}
The proof is fully analogous to that of~\bref{prop:gran-utv-Lambda0}. As an overcomplete coordinate system on the above sub-gPDE we can take the restrictions of : $\{\xi^A,\rho_{A}{}^B, \lambda,\lambda^A, \We^{m}{}_{nkp;(a)}, |a|\geq 0\}$. In these coordinates the action of $Q$  on the degree $1$ coordinates reads as:
\begin{align}\label{gran-ghostslambdaneq0}
    \begin{split}
        &Q\xi^A=\xi^B \rho_{B}{}^{A}-\xi^A\lambda,\\
        &Q \rho_{A}{}^B=\rho_{A}{}^C \rho_{C}{}^B+\lambda_A\xi^B-\lambda^B \xi_A +\dfrac{1}{2}\xi^C\xi^D \We^{B}{}_{ACD},\\
        &Q\lambda=\xi^A\lambda_A,\\
        &Q\lambda^A=\rho^{A}{}_C\lambda^C-\lambda\lambda^A+\dfrac{1}{2}\xi^B \xi^C \Co^{A}{}_{BC}\,.
    \end{split}
\end{align}
It is easy to see that this coincides with the definition of CE differential of $o(d,1)$ for dS and $o(d-1, 2)$ for AdS, 
written in the conformal-like basis.  This of course signals that in the case at hand the boundary is naturally equipped with the conformal structure. More precisely, solutions to the above sub-gPDE in the sector of degree $1$ coordinates define a Cartan connection of the respective conformal geometry. However, the gauge theory encoded in this sub-gPDE  is not generally equivalent to conformal geometry. For instance, in the case of $d=5$ the respective conformal geometry is Bach-flat. The Bach flatness condition is encoded in the equations on curvatures arising in the sector of degree $0$ variables. This is the realization in our approach of the well-known Fefferman-Graham analysis~\cite{Fefferman-Graham:1985ambient,Fefferman:2007rka} (see also~\cite{Bekaert:2012vt,Bekaert:2013zya,Chekmenev:2015kzf,Bekaert:2017bpy} for the analogous considerations for generic gauge fields within a version of gPDE framework).

As for asymptotic symmetries, one can consider an analogous boundary condition $\sigma^*(\xi^A)=e^A$, where $e^A$ is a fixed frame on the boundary. The analysis of Section~\bref{sec:eom-min} can be easily repeated in the case at hand, giving the conformal Killings  of $g_{\mu\nu}= e_\mu^A e_\nu^B g_{AB}$ as the basis in the algebra of asymptotic symmetries. Of course, one can equally well repeat the analysis of Section~\bref{sec:Bound+BMS} in which case together with $\xi^A-e^A\sim 0$ and $Q(\xi^A-e^A)\sim 0$ one should also impose additional boundary conditions $\lambda\sim 0$ and $Q\lambda\sim 0$. 

\section*{Acknowledgments}
\label{sec:Aknowledgements}
We wish to thank I.~Dneprov and Th.~Popelensky for fruitful discussions. M.G. is also grateful to G.~Barnich, X.~Bekaert, M.~Henneaux, and J.~Herfray for useful exchanges.
The work of M.~M. was supported by the Russian Science Foundation grant No 22-72-10122 (\url{https://rscf.ru/en/project/22-72-10122/}).
Part of this work was done when M.~G. participated in the thematic program "Emergent Geometries from Strings and Quantum Fields" at the Galileo Galilei Institute for Theoretical Physics, Florence, Italy.

\appendix
\section{Symmetries}
Let us restrict ourselves to local analysis. At least locally, a gPDE can be equivalently represented as the nonlagrangian local BV system so that it is enough to give  a proof in this setup. In this case $E$ is a $J^\infty(\cE) \to X$ pulled back to $T[1]X$. In particular, functions on $E$ can be identified with horizontal forms on~$J^\infty(\cE)$. $E$ is equipped with the evolutionary homological vector filed $s$ of ghost degree $1$ and the homological vector field $\dh=\theta^aD_a$ of $\theta$-homogeneity $1$. We have the following:
\begin{prop}
Locally, cohomology of $\commut{\dh}{\cdot}$ in the space of vertical vector fields is trivial in positive $\theta$-homogeneity (space-time form degree). In the space of vertical vector fields of vanishing $\theta$-homogeneity, it is given by the evolutionary vector fields on $E$.
\end{prop}
\begin{proof}
Let us work in local coordinates $x^a,\theta^b,\phi^i_{(a)}$ and conisider a $\theta$-homogeneity 1 vector field to begin with: 
\begin{equation}
V=\theta^a V_a^i\dl{\phi^i}+\theta^a V_{a|b}^i\dl{\phi^i_b}+\ldots    \,.
\end{equation}
The cocycle condition reads as $\commut{\dh}{V}=0$ and implies, in particular,
\begin{equation}
\commut{\dh}{V}\phi^i=\theta^a\theta^b(D_b V^i_a-V^i_{a|b})=0
\end{equation}
At the same time the coboundary $\commut{\dh}{W}$  acting $\phi^i$ has the following structure:
\begin{equation}
\commut{\dh}{W}\phi^i=\theta^a (D_a W^i-W^i_a)\,,
\end{equation}
where coefficients $W^i, W^i_a$ are introduced as $W^i=W {\phi^i}$ and $W^i_a=W {\phi^i_a}$. 
It follows, by adding a coboundary one can always set the coefficient  $V^i_a=\dl{\theta^a}(V\phi^i)$ to zero. Then the cocycle condition implies that $V^i_{a|b}-V^i_{b|a}=0$ so that 
$V^i_{b|a}$ can be also set to zero by adding $\commut{\dh}{W}$ such that the only nonvanishing coefficient is $W^i_{ab}=W\phi^i_{(ab)}=-V^i_{(a|b)}$. The proof can be completed by induction.  The analysis for higher $\theta$-homogeneity vector field is analogous.    
\end{proof}

Let us now turn to the cohomology of $\commut{Q}{\cdot}$, $Q=\dh+s$ in the space of vertical vector fields. Expanding the cocycle condition in the $\theta$-homogeneity one gets
\begin{equation}
\begin{gathered}
\commut{s}{V_0}=0\,, \quad \commut{s}{V_1}+\commut{\dh}{V_0}=0\,, \quad \ldots \\ 
\commut{s}{V_1}+\commut{\dh}{V_0}=0\,, \quad \commut{\dh}{V_k}=0\,,
\end{gathered}
\end{equation}
where we assumed that $V_l=0$ for all $l>k$, with $0<k \leq n$. Applying the above Proposition we conclude that $V_l=\commut{\dh}{W_{l-1}}$. Subtracting the trivial cocycle $\commut{Q}{W_{l-1}}$ we arrive at a new representative $V^\prime$ such that $V^\prime_l=0$ for all $l>k-1$. Applying the same procedure again we arrive at an equivalent representative for which $V_l=0$ for $l>0$. The cocycle condition then implies that $\commut{s}{V_0}=0$ and 
$\commut{\dh}{V_0}=0$. In other words we have arrived at the standard representative of a global symmetry. Note that in the above we did not assume $\gh{V}=0$ so that it applies to generalized symmetries as well.

\label{sec:gpdesym}

\setlength{\itemsep}{1pt}
\small

\providecommand{\href}[2]{#2}\begingroup\raggedright\endgroup

\end{document}